\documentclass[3p,12pt,authoryear]{elsarticle}
\biboptions{semicolon}
\usepackage{lineno}
\usepackage{caption}
\usepackage{subcaption}	
\usepackage[shortlabels]{enumitem}

\usepackage{amsthm}
\usepackage{amsmath,amssymb}
\usepackage{mathtools}
\usepackage{chngcntr}
\usepackage{url}

\usepackage{algorithmicx}
\usepackage{algpseudocode}
\usepackage{algorithm}
\usepackage{hyperref}

\newcommand{\OPT}{{\mathop{\rm OPT}}}
\newtheorem{theorem}{Theorem}[section]
\newtheorem{lemma}{Lemma}[section]
\newtheorem{prop}{Proposition}[section]
\newtheorem{corollary}{Corollary}[section]
\newtheorem{observation}{Observation}[section]
\newtheorem{property}{Property}[section]

\begin{document}

\begin{frontmatter}

\title{Algorithms and Computational Study on a Transportation System Integrating Public Transit and Ridesharing of Personal Vehicles\tnoteref{t1}}

\tnotetext[t1]{A preliminary version of the paper appeared in the Proceedings of ISAAC2021~\citep{Gu-ISAAC21}.}

\author[1]{Qian-Ping Gu}
\ead{qgu@sfu.ca}

\author[1]{Jiajian Leo Liang\corref{cor1}}
\ead{leo\_liang@sfu.ca}

\affiliation[1]{organization={School of Computing Science, Simon Fraser University},
            city={Burnaby},
            state={British Columbia},
            country={Canada}}

\cortext[cor1]{Corresponding author}

\begin{abstract}
The potential of integrating public transit with ridesharing includes shorter travel time for commuters and higher occupancy rate of personal vehicles and public transit ridership.
In this paper, we describe a centralized transit system that integrates public transit and ridesharing to reduce travel time for commuters.
In the system, a set of ridesharing providers (drivers) and a set of public transit riders are received. 
The optimization goal of the system is to assign riders to drivers by arranging public transit and ridesharing combined routes subject to shorter commuting time for as many riders as possible.
We give an exact algorithm, which is an ILP formulation based on a hypergraph representation of the problem.
By using the ILP and the hypergraph, we give approximation algorithms based on LP-rounding and hypergraph matching/weighted set packing, respectively.
As a case study, we conduct an extensive computational study based on real-world public transit dataset and ridesharing dataset in Chicago city.
To evaluate the effectiveness of the transit system and our algorithms, we generate data instances from the datasets.
The experimental results show that more than 60\% of riders are assigned to drivers on average, riders' commuting time is reduced by 23\% and vehicle occupancy rate is improved to almost 3.
Our proposed algorithms are efficient for practical scenarios.
\end{abstract}


\end{frontmatter}

\section{Introduction} \label{sec-intro}
As the population grows in urban areas, commuting between and within large cities is time-consuming and resource-demanding.
Due to growing passenger demand, the number of vehicles on the road for both public and private transportation has increased to handle the demand.
Current public transportation systems in many cities may not be able to handle the increasing demand due to their slow (or lack of) transit development.
This can cause greater inconvenience for transit users, such as longer waiting time, more transfers and/or imbalanced transit ridership.
As a result of the inconvenience, many people choose to use personal vehicles for work commute.
According to~\citep{CSS20, Sierpinski-ATST13, StatsCan-2016}, personal vehicles were the main transportation mode in the United States and Canada and in more than 200 European cities.
Almost 73\% of total work commute is by car as a driver in Canada in recent years~\citep{StatsCan-2016}.
In the US, the growth rate of population actually has become higher than that of the transit ridership since 2016 (9\% higher in 2019) according to~\citet{APTA22}.
The occupancy rate of personal vehicles in the U.S. was 1.6 persons per vehicle in 2011~\citep{Ghoseiri-USDT11,Santos-USDTFHA11} and decreased to 1.5 persons per vehicle in 2017~\citep{CSS20}, which can be a major cause for congestion and pollution.
This is the reason municipal governments encourage the use of public transit.

One of the major drawbacks of public transit is the inconvenience and inflexibility of first mile and last mile (\textbf{FM/LM}) transportation, compared to personal vehicles~\citep{Burstlein-TRAPP21,Chen-TRBM18,Wang-TS14}.
The FM transportation refers to the transit service from a passenger's home/origin to the nearest public transportation hub, and vice versa for LM transportation.
Mobility-on-demand (MoD) systems, such  as Uber, Lyft and DiDi, have become popular around the globe for their convenience.
With the increasing popularity in ridesharing/ridehailing service, there may be potential in integrating private and public transportation systems, as suggested by some studies~\citep[e.g.,][]{Alonso-Gonzalez-TRR18,Huang-TITS19,Kumar-TRCET21,Ma-TRELTR19,Narayan-TRCET20,Stiglic-COR18}.
However, the under utilization of MoD vehicles may have increased congestion and $\textup{CO}_2$ emissions due to the increased of low-occupancy ridesharing vehicles on the road~\citep{Diao-NS21, Henao-Trans19,Tirachini-IJST20}.
From the research report of~\citet{Feigon-TRB16}, it is recommended that public transit agencies should build on mobility innovations to allow public-private engagement in ridesharing because the use of shared modes increases the likelihood of using public transit.
A similar finding, reported in~\citet{Zhang-IJERPH18}, indicates that the use of ridesharing may be positively associated with public transit ridership.
As pointed out by~\citet{Ma-TRELTR19}, some basic form of collaboration between MoD services and public transit already exists for FM/LM transportation.
For example, \citet{Thao-TRIP21} mentioned in their study that a basic integration of ridesharing and public transport in rural Switzerland is already in place; and
there have been pilot projects (e.g., Taxito, Ebuxi/mybuxi, Kollibri, sowiduu) promoting the integration of public transit and MoD in Switzerland.
There is an increasing interest for collaboration between private companies and public sector entities~\citep{Raghunathan-SMT18}.

The spareness of transit networks usually is the main cause of the inconvenience in public transit.
Such transit networks have infrequent transit schedule and can cause customers to have multiple transfers.
In this paper, we investigate the potential effectiveness of integrating public transit with ridesharing to reduce travel time for commuters and increase occupancy rate in such sparse transit networks (with the focus on work commute).
For example, people who drive their vehicles to work can pick-up \textit{riders}, who use public transit regularly, at designated locations and drop-off them at some transit stops, and then these riders can take public transit to their destinations.
In this way, riders are presented with a cheaper alternative than ridesharing for the entire trip, and it is more convenient than using public transit only.
The transit system also gets a higher ridership, which aligns with the recommendation of~\citet{Feigon-TRB16} and \citet{Zhang-IJERPH18} for a more sustainable transportation system.

Our research focuses on a centralized transit system that is capable of matching drivers and riders satisfying their trips' requirements while achieving some optimization goal. The requirements of a trip may include an origin and a destination, time constraints, capacity of a vehicle, and so on.
For a rider, a \emph{ridesharing route} contains a road segment where the rider is served by a driver and a road segment where the rider uses public transit; and a \emph{public transit route} means that the rider uses public transit only.
To improve commuting time for riders and increase transit ridership, the transit system's optimization goal is to maximize the number of riders, each of whom is assigned a ridesharing route that is quicker than the fastest public transit route for the rider. We call this the multimodal transportation with ridesharing (\textbf{MTR}) problem (formal definition is given in Section~\ref{sec-preliminary}).

\subsection{Related work}\label{subsec-related-work}
There is a rich literature on standalone ridesharing and carpooling, from theoretical to computational studies~\citep[some representative works][]{Agatz-TRBM11,Alonso-Mora-PNAS17,Gu-TCS21}, and recent literature reviews can be found in~\citep{Mourad-TRBM19,Tafreshian-SS20,Wang-TRBM19}.
A few studies on the integration of public transit with dynamic ridesharing have also been reported. 
\citet{Aissat-ICEIS15} proposed an approach which, given a public transit route for a rider, substitutes each part of the route with ridesharing if ridesharing is better than the original part. 
Their algorithm finds the best route for each rider in first-come first-serve (FCFS) basis, where an optimization goal of the system is not considered, and the algorithm is computational intensive.
\citet{Huang-TITS19} presented a more robust approach, compared to~\citet{Aissat-ICEIS15}, by combining two networks $N, N'$ (representing the public transit and ridesharing networks, respectively) into one single routable graph $G$.
The graph $G$ uses the \emph{time-expanded model} to maintain the information about all public-vehicle schedule, riders' and drivers' origins, destinations and time constraints.
For any rider travel query, a public transit and ridesharing combined route (ridesharing route as defined above) is found on $G$, if available, by a shortest path algorithm.
Their approach is also FCFS basis.
\citet{Masoud-TRR17} used a similar idea of time-expanded network.
Their algorithm and experiment consider only a limited number of transfer points, which may constrain their algorithm in large-scale transit systems.
Due to the nature of FCFS basis in the above mentioned papers, no exact and approximation algorithms are considered to achieve an overall optimization goal.

\citet{Kumar-TRCET21} studied the FM/LM problem for transit in which individuals have no or limited transit service due to limited transit coverage and connectivity.
The authors use a schedule-based transit network graph to determine a set of feasible matches (a feasible match has a driver and a rider satisfying all constraints). 
Then, they formulate an integer linear programming (ILP) to find an optimal matching from the computed feasible matches.
Their approach only consider at most one rider per driver.
\citet{Molenbruch-EJOR21} studied a similar FM/LM problem where public transit is not available to riders, so the riders have to rely on demand-responsive services to connect to major transit stations. The difference is that~\citet{Molenbruch-EJOR21} focuses more on the dial-a-ride problem (DARP)~\citep{Cordeau-TRBM03}.
The authors use a variant of the metaheuristic Large Neighborhood Search to compute a list of candidate riders that can be served by each DAR vehicle.
Then, candidates are verified based on their (users' and DAR vehicles') time constraints.
Objective is to design minimum-distance DAR routes, satisfying all user requests.

\citet{Luo-JTL21} proposed a different multimodal transportation system for the FM/LM problem.
Their transportation system integrates micromobility services (including bikes and electric scooters sharing) with MoD.
In the system, travellers use micromobility services for the FM/LM connections to hubs for ridesharing supported by MoD.
The optimization goals include finding optimum placements of hubs and micromobility vehicles for given demands and supplies (re-positioning the micromobility vehicles as well).
\citet{Salazar-TITS20} introduced autonomous vehicle into the integration of public transit and MoD while considering the energy consumption MoD autonomous vehicles.
A relevant optimization goal considered in~\citet{Salazar-TITS20} is to minimize riders' travel time  together with the operational costs of the autonomous fleet and public transportation.

Several transportation systems integrating public transit and ridesharing have been proposed to address the FM/LM problem.
\citet{Ma-TRELTR19} presented an integrated transportation system to provide ridesharing services using a fleet of dedicated vehicles. The system is FCFS basis.
When a rider request enters the system, a fastest travel option is computed approximately and provided to the rider.
The system also includes the relocation of ridesharing vehicles if they are idle.
\citet{Narayan-TRCET20} gave a similar system (without vehicle re-location).
Their conditions for best routes are different from~\citet{Ma-TRELTR19}, and the routes are computed differently.
The models in~\citet{Ma-TRELTR19} and \citet{Narayan-TRCET20} are different from our model in which personal drivers are the main focus.

\citet{Ma-EEEIC17} and \citet{Stiglic-COR18} proposed models to integrate public transit and ridesharing as graph matching problems to achieve certain optimization goals.
In the models of~\citet{Ma-EEEIC17} and \citet{Stiglic-COR18}, a set of drivers, a set of riders and a public transit network are given.
Their models assign riders to drivers to use ridesharing for replacing FM/LM transit (exclusive ridesharing can be supported, as described in~\citet{Stiglic-COR18}).
A group of riders can be assigned to a driver if all constraints of the riders in the group and the driver are satisfied.
Each of rider groups and drivers is represented as a node in a shareability graph (RV graph~\citet{Santi-PNAS14} and RTV graph~\citet{Alonso-Mora-PNAS17}); and there is an edge between a rider group node and a driver node if the group of riders can be assigned to the driver. The rider and driver assignment problem is modelled as graph matching problem in the graph (then formulated as an ILP problem and solved by an ILP solver). 
The optimization goal in~\citet{Ma-EEEIC17} is to minimize the cost related to waiting time and travel time. Two optimization goals are considered in~\citet{Stiglic-COR18}: one is to maximize the number of riders assigned to drivers, and the other is to minimize the total distance increase for all drivers. These models are closely related to ours, but there are differences and limitations.
The goal in~\citet{Ma-EEEIC17} is different from ours and does not guarantee ridesharing routes better than the public transit routes.
One of the goals in~\citet{Stiglic-COR18} aligns with ours, but there are restrictions: a rider group can have at most two riders, each rider must use the transit stop closest to the rider's destination, and more importantly, the ridesharing routes are not guaranteed to be better than public transit routes. 

It is worth to mention that the ILP formulation for the MTR problem (described in Section~\ref{sec-exact-IP}) is a special case of the Separable Assignment Problem (SAP), which is a generalization of the Generalized Assignment Problem (GAP).
Given a $\beta$-approximation algorithm for the single-bin subproblem in SAP,~\citet{Fleischer-SODA06} presented two approximation algorithms for SAP: an LP-rounding based $((1-\frac{1}{e})\beta)$-approximation algorithm and a local-search $(\frac{\beta}{\beta+1}-\epsilon)$-approximation algorithm, $\epsilon>0$.
If interested, a problem related to SAP is the Multiple Knapsack Problem with Assignment Restrictions~\citep[e.g.,][]{Dawande-JCO00}.

\subsection{Contribution}
In this paper, we propose a system integrating public transit and ridesharing to provide ridesharing routes quicker than public transit routes for as many riders as possible (MTR problem). 
We use a similar model as in~\citet{Ma-EEEIC17} and \citet{Stiglic-COR18} to solve the MTR problem; and we extend the work in~\citet{Stiglic-COR18} to eliminate the limitations described above. That is, a ridesharing match allows more than two riders, riders can be picked-up/dropped-off at any feasible transit stop, and ridesharing routes assigned to riders are quicker than the fastest public transit routes.
Also, no approximation algorithms are presented in~\citet{Ma-EEEIC17} and \citet{Stiglic-COR18}, and
we give an exact algorithm and approximation algorithms for the optimization problem to ensure solution quality.
Our exact algorithm approach is also similar to the approach proposed in~\citet{Alonso-Mora-PNAS17} and \citet{Santi-PNAS14} for ridesharing.
No approximation algorithm is presented in~\citet{Alonso-Mora-PNAS17}, whereas
\citet{Santi-PNAS14} obtained a $\frac{1}{K+1}$-approximation algorithm assuming the maximum capacity $K$ among all vehicles (drivers) is at least two.
We prove that our approximation algorithms for the MTR problem have constant approximation ratios, regardless of the vehicle capacity.
Our discrete algorithms allow to control the trade-off between quality and computational time.
We evaluate the proposed exact algorithm and approximation algorithms by conducting an extensive computational study based on real-life data.
Our main contributions are summarized as follows:
\begin{enumerate}
\setlength\itemsep{0em}
\item We give an exact algorithm approach (an ILP formulation based on a hypergraph representation) for integrating public transit and ridesharing.

\item We prove the MTR problem is NP-hard. We give an LP-rounding based $(1-\frac{1}{e})$-approximation algorithm and a discrete $\frac{1}{2}$-approximation algorithm for the problem. We show that previous approximation algorithms in~\citet{Berman-SWAT00} and \citet{Chandra-JoA01} for the $k$-set packing problem are $\frac{1}{2}$-approximation algorithms for the MTR problem. Our $\frac{1}{2}$-approximation algorithm is more time and space efficient than the $k$-set packing algorithms.

\item Based on real-life data in Chicago City, we create data instances for empirical study to evaluate the potential of an integrated transit system and the effectiveness of the exact algorithm and approximation algorithms. We conduct extensive experiments on the data instances; and the results show that our integrated transit system is able to assign more than 60\% of riders to drivers, reduce riders' travel time by 23\% and increase the occupancy rate of personal vehicles to about three.
\end{enumerate}

The rest of the paper is organized as follows.
In Section~\ref{sec-preliminary}, we give the preliminaries of the paper, describe a centralized system that integrates public transit and ridesharing, and define the MTR optimization problem.
In Section~\ref{sec-exact}, we describe our exact algorithm approach. We then propose approximation algorithms in Section~\ref{sec-approximate}.
We discuss our numerical experiments and results in Section~\ref{sec-experiment}.
Finally, Section~\ref{sec-conclusion} concludes the paper.

\section{Problem definition and preliminaries} \label{sec-preliminary}
In the \textit{multimodal transportation with ridesharing} (\textit{MTR}) problem, we have a centralized system, and for every fixed time interval, the system receives a set $\mathcal{A} = D \cup R$ of participant trips with $D \cap R = \emptyset$, where $D$ is the set of driver trips and $R$ is the set of rider trips.
Each trip in $\mathcal{A}$ is expressed by an integer label $i$, so an integer labeled trip $i \in \mathcal{A}$ may be referred to as a driver trip or a rider trip.
Each trip consists of an individual (driver or rider), a vehicle (for driver trip only) and some requirements.
For brevity, we usually call a driver trip just a \textit{driver} and a rider trip a \textit{rider}.
A connected public transit network with a fixed timetable $T$ is given.
The timetable $T$ contains each transit vehicle's departure and arrival times for each transit stop/station (a transit vehicle includes bus, metro train, rail and so on).
We assume that given an earliest departure time from any source $o$ and a destination $d$ in the public transit network, the fastest travel time from $o$ to $d$ (including transfer time) can be computed from $T$ quickly.
Given an earliest departure time $dt$ chosen by a rider $i \in R$, a \emph{public transit route} $\hat{\pi}_i(dt)$ for rider $i$ is a travel plan using only public transportation, whereas a \emph{ridesharing route} $\pi_i(dt)$ for rider $i$ is a travel plan using a combination of public transportation and ridesharing to reach $i$'s destination satisfying $i$'s requirements.
The multimodal transportation with ridesharing (MTR) problem asks to provide at least one feasible route ($\pi_i(dt)$ or $\hat{\pi}_i(dt)$) for every rider $i \in R$. We denote an instance of the MTR problem by $(N,\mathcal{A},T)$, where $N$ is an edge-weighted directed graph (road network) that represents locations reachable by both private and public transportation.
We call a public transit station or stop just \emph{station}.
The terms rider and passenger are used interchangeably (although passenger emphasizes a rider who has accepted a ridesharing route).

The requirements of each trip $i$ in $\mathcal{A}$ are specified by $i$'s parameters submitted by the individual.
The parameters of a trip $i$ contain an origin location $o_i$, a destination location $d_i$, an earliest departure time $\alpha_i$, a latest arrival time $\beta_i$ and a maximum trip time.
A driver trip $i$ also contains a capacity of the vehicle, a limit on the number of stops a driver wants to make to pick-up/drop-off passengers, and an optional path to reach its destination.
The maximum trip time of a driver $i$ includes a travel time from $o_i$ to $d_i$ and a detour time limit $i$ can spend for offering ridesharing service.
A rider trip $i$ also contains an acceptance threshold $\theta_i$ for a ridesharing route $\pi_i(\alpha_i)$, that is, $\pi_i(\alpha_i)$ is given to rider $i$ if $t(\pi_i(\alpha_i)) \leq \theta_i \cdot t(\hat{\pi}_i(\alpha_i))$ for every public transit route $\hat{\pi}_i(\alpha_i)$ and $0 < \theta_i \leq 1$, where $t(\cdot)$ is the travel time of a route.
Such a route $\pi_i(\alpha_i)$ is called an \emph{acceptable ridesharing route} (acceptable route for brevity).
For example, suppose the fastest public transit route $\hat{\pi}_i(\alpha_i)$ takes 100 minutes for $i$ and $\theta_i = 0.9$. An acceptable route $\pi_i(\alpha_i)$ implies that $t(\pi_i(\alpha_i)) \leq \theta_i \cdot t(\hat{\pi}_i(\alpha_i)) = 90$ minutes.
We consider two match types for practical reasons (although our system can extend to support different match types).
\begin{itemize}
\item \textbf{Type 1 (rideshare-transit)}: a driver may make multiple stops to pick-up different passengers, but makes only one stop to drop-off all passengers. In this case, the \emph{pick-up locations} are the passengers' origin locations, and the \emph{drop-off location} is a public station.

\item \textbf{Type 2 (transit-rideshare)}: a driver makes only one stop to pick-up passengers and may make multiple stops to drop-off all passengers. In this case, the \emph{pick-up location} is a public station and the \emph{drop-off locations} are the passengers' destination locations.
\end{itemize}
A rider $j$ is said to be \emph{served} by a driver $i$ if $j$ agrees to be picked-up and dropped-off by $i$ following an acceptable route $\pi_j(\alpha_j)$ from $i$.
Riders and drivers specify which match type to participate in; they are allowed to choose both in hope to increase the chance being selected, but the system will assign them only one of the match types such that the optimization goal of the MTR problem is achieved, which is to assign acceptable routes to as many riders as possible.
Formally, the optimization goal of the MTR problem is to maximize the number of passengers, each of whom is assigned an acceptable route $\pi_i(\alpha_i)$ for every $i \in R$; and we simply call this optimization problem the \textbf{MTR problem} hereafter.

For a driver $i$ and a set $J \subseteq R$ of riders, the set $\sigma(i) = \{i\} \cup J$ is called a \emph{feasible match} if driver $i$ can serve this group $J$ of riders together, using a route in $N$ from $o_i$ to $d_i$, while all requirements (i.e., constraints) specified by the parameters of the trips in $\sigma(i)$ are satisfied collectively. The construction of a feasible match $\sigma(i)$ and its feasibility check, related to the constraints, are described in detail in Section~\ref{sec-compute-matches}.
Two feasible matches $\sigma(i)$ and $\sigma(i')$ are \emph{disjoint} if $\sigma(i) \cap \sigma(i') = \emptyset$.
Then, a solution to the MTR problem is a set of pairwise disjoint feasible matches such that the total number of passengers included in the feasible matches is maximized.
In any solution to the MTR problem, each driver $i$ belongs to at most one feasible match, $i$ serves exactly one group of riders, and any rider $j$ belongs to at most one served group (feasible match).

Intuitively, a rideshare-transit (Type 1) feasible match $\sigma(i)$ is that all passengers in $\sigma(i)$ are picked-up at their origins and dropped-off at a station, and then $i$ drives to destination $d_i$ while each passenger $j$ of $\sigma(i)$ takes public transit to destination $d_j$.
A transit-rideshare (Type 2) feasible match $\sigma(i)$ is that all passengers in $\sigma(i)$ are picked-up at a station and dropped-off at their destinations, and then $i$ drives to destination $d_i$ after dropping-off the last passenger.
We give algorithms to find pairwise disjoint feasible matches to maximize the number of passengers included in the matches.
We describe our algorithms for Type 1 only. Algorithms for Type 2 can be described with the constraints on the drop-off location and pick-up location of a driver exchanged, and we omit the description.
Further, it is not difficult to extend to other match types, such as ridesharing only and park-and-ride, as described in~\citep{Stiglic-COR18}.

\section{Exact algorithm approach} \label{sec-exact}
Our exact algorithm approach for the MTR problem is presented in this section, which is similar to the matching approach described in~\citet{Alonso-Mora-PNAS17} and \citet{Santi-PNAS14} for ridesharing and in~\citet{Ma-EEEIC17} and \citet{Stiglic-COR18} for integration of public transit and ridesharing.
Our approach computes a hypergraph representing all feasible matches in an instance $(N, \mathcal{A},T)$ of the MTR problem, which is similar to \cite{Luo-arXiv22} and \citet{Santi-PNAS14}.
The hypergraph allows a better intuition of the approximation algorithms described in Section~\ref{sec-approximate}.
We give a detailed description of our hypergraph approach as some of the definitions/notations are used in our approximation algorithms.
More importantly, we show that the MTR problem is in fact NP-hard using our ILP formulation from the hypergraph representing.

\subsection{Integer program formulation} \label{sec-exact-IP}
The exact algorithm approach is described in the following.
Given an instance $(N, \mathcal{A},T)$ of the MTR problem, we first compute all feasible matches for each driver $i \in D$.
Then, we create a bipartite (hyper)graph $H(V,E)$, where $V(H) = D\cup R$.
For each driver $i \in D$ and a non-empty subset $J \subseteq R$, if $\{i\} \cup J$ is a feasible match, create a hyperedge $e = (i, J)$ in $E(H)$.
Any driver $i \in D$ or rider $j\in R$ does not belong to any feasible match is removed from $V(H)$, that is, $H$ contains no isolated vertex (such riders must use public transit routes).
For each edge $e=(i,J)$ in $E(H)$, assign a weight $w(e) = |J|$ (representing the number of riders in $e$).
Let $D(H) = D\cap V(H)$ and $R(H) = R \cap V(H)$.
For an edge $e = (i, J)$ in $E(H)$, let $D(e)=i$ (the driver of $e$) and $R(e)=J$ (the passengers of $e$).
For a subset $E' \subseteq E(H)$, let $D(E')= \cup_{e\in E'} D(e)$ and $R(E') = \cup_{e\in E'} R(e)$.
For a vertex $j \in V(H)$, define $E_j = \{e \in E(H) \mid j \in R(e)\}$ to be the set of edges in $E(H)$ incident to $j$.
An example of the hypergraph $H(V,E)$ is given in Figure~\ref{fig-hypergraph}.
\begin{figure}[b]
\centering
\includegraphics[width=.7\linewidth]{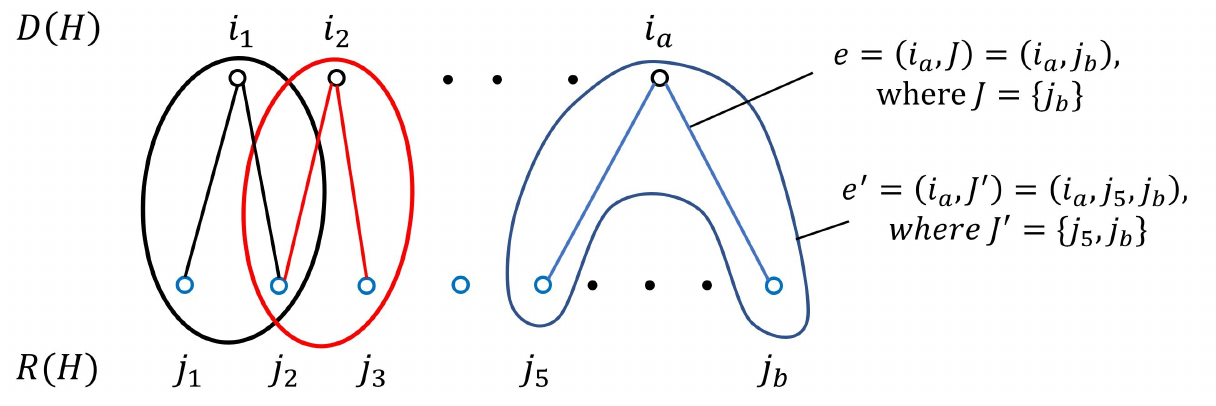}
\caption{A bipartite hypergraph $H(V,E)$ representing all feasible matches of an instance $(N,\mathcal{A},T)$, where $|D(H)|=a$ and $|R(H)|=b$.}
\label{fig-hypergraph}
\end{figure}
To solve the MTR problem, we give an integer linear programming (ILP) formulation:
\begin{alignat}{4}
 & \text{maximize }   &        &\sum_{e \in E(H)} w(e) \cdot x_{e} & \qquad \label{obj-1}\\
 & \text{subject to } & \qquad &\sum_{e \in E_j} x_{e} \leq 1,   & & \forall \text{ } j \in \mathcal{A} \label{constraint-1}\\
 &                    &        &x_{e} \in \{0,1\}, & &\forall \text{ } e \in E(H) \label{constraint-2}
\end{alignat}
The binary variable $x_e$ indicates whether the edge $e = (i, J)$ is in the solution ($x_e = 1$) or not ($x_e = 0$).
If $x_e = 1$, it means that all passengers of $J$ are assigned to $i$ and can be delivered by $i$ satisfying all constraints.
Inequality (2) in the ILP formulation guarantees that each driver serves at most one feasible set of passengers and each passenger is served by at most one driver.
Note that the ILP \eqref{obj-1}-\eqref{constraint-2} is similar to a set packing formulation.
An advantage of this ILP formulation is that the number of constraints is substantially decreased, compared to traditional ridesharing formulation.

\begin{observation}
A match $\sigma(i)$ for any driver $i \in D$ is feasible if and only if for every subset $J \subseteq (\sigma(i)\setminus\{i\})$, $\{i\} \cup J$ is a feasible match~\citep{Stiglic-TRBM15}.
\label{obs-base-match}
\end{observation}

From Observation~\ref{obs-base-match}, it is not difficult to see that Proposition~\ref{prop-1} holds.
\begin{prop}
Let $D' \subseteq D$ and $P(D')$ be a maximal set of passengers served by $D'$.
There always exists a set of feasible matches for $D'$ such that $\sigma(i) \cap \sigma(i') = \emptyset$ for every $i, i' \in D'$ and $\bigcup_{i \in D'} \sigma(i) \setminus \{i\} = P(D')$.
\label{prop-1}
\end{prop}

\begin{theorem}\label{theorem-ILP}
Given a hypergraph $H(V,E)$ for an instance of the multimodal transportation with ridesharing (MTR) problem, 
an optimal solution to the ILP (\ref{obj-1})-(\ref{constraint-2}) is an optimal solution to the MTR problem and vice versa.
\end{theorem}

\begin{proof}
From inequality (\ref{constraint-1}) in the integer program, the solution found by the integer program is always feasible to the MTR problem.
By Proposition~\ref{prop-1} and objective function~(\ref{obj-1}), an optimal solution to the ILP (\ref{obj-1})-(\ref{constraint-2}) is an optimal solution to the MTR problem.
Obviously, an optimal solution to the MTR problem is an optimal solution to the ILP (\ref{obj-1})-(\ref{constraint-2}).
\end{proof}

\subsection{Computing feasible matches} \label{sec-compute-matches}
Let $i$ be a driver in $D$ and $n_i$ be the capacity of $i$ (the maximum number of riders $i$ can serve at once). The maximum number of feasible matches for a single driver $i$ is $\sum_{p = 1}^{n_i} \binom{|R|}{p}$.
Assuming the capacity $n_i$ is a very small constant (which is reasonable in practice), the above summation is polynomial in $R$, that is, $O((|R|+1)^{n_i})$ (partial sums of binomial coefficients).
Let $K = \max_{i \in D} {n_i}$ be the maximum capacity among all vehicles/drivers.
Then, in the worst case, $|E(H)| = O(|D| \cdot (|R|+1)^K)$.

In the following, we describe how to compute all feasible matches between drivers and riders in $\mathcal{A}=D\cup R$, given an instance $(N, \mathcal{A}, T)$ of the MTR problem.
During the computation of a feasible match $\sigma(i)=\{i\}\cup J$ for some driver $i \in D$ and subset $J \subseteq R$, the route $r_i$ (actual travel path) of shortest travel time is computed for $i$ to serve every one of $\sigma(i)$, assuming a shortest path from one location to another can be computed from roadmap network $N$.
Because the number of locations $i$ needs to visit for a feasible match $\sigma(i)$ is limited, enumerating all possible locations to compute $r_i$ is still quick (detailed description is given in Section~\ref{subsection-alg-feas-all}).
The general procedure to compute all feasible matches is similar to~\citet{Stiglic-COR18} with some minor differences to further extend and overcome the limitations of~\citet{Stiglic-COR18}, as mentioned in the related work (Section~\ref{subsec-related-work}).
Further, some definitions are required to show that the route $r_i$ computed for driver $i$ indeed has the shortest travel time.
Hence, we give a full description for computing all feasible matches.

Computing all feasible matches between $D$ and $R$ is done in two phases.
In phase one, for each driver $i$, we find all feasible matches $\sigma(i)=\{i,j\}$ with one rider $j$.
In phase two, for each driver $i$, we compute all feasible matches $\sigma(i)=\{i,j_1,..,j_p\}$ with $p$ riders, based on the previously computed feasible matches $\sigma(i)$ with $p-1$ riders, for $p=2$ upto the number of riders $i$ can serve.
Before describing how to compute the feasible matches, we first introduce some notations and specify the feasible match constraints we consider in more detail.
Each trip $i \in \mathcal{A}$ is specified by the parameters $(o_i, d_i, n_i, z_i, p_i, \delta_i, \alpha_i, \beta_i, \gamma_i, \theta_i)$, where the parameters are summarized in Table~\ref{table-notation} along with other notation.
\begin{table}[!tp]
\footnotesize
\setlength\tabcolsep{4pt}
\centering
   \begin{tabular}{c | l}
   	\hline
   	\textbf{Notation}  & \textbf{Definition}                                         \\ \hline
   	$o_i$              & Origin (start location) of $i$ (a vertex in $N$)             \\
   	$d_i$              & Destination of $i$ (a vertex in $N$)                         \\
   	$n_i$              & Number of seats (capacity) of $i$ available for passengers (driver only)   \\
   	$z_i$              & Maximum detour time driver $i$ willing to spend for offering ridesharing services    \\
   	$p_i$       	    & An optional preferred path of $i$ from $o_i$ to $d_i$ in $N$ (driver only)     	\\
   	$\delta_i$         & Maximum number of stops $i$ willing to make to pick-up passengers for match     \\
   	                   & Type 1 and to drop-off passengers for match Type 2 \\
   	$\alpha_i$         & Earliest departure time of $i$                                     	\\
   	$\beta_i$          & Latest arrival time of $i$                                         	\\
    $\gamma_i$         & Maximum trip time of $i$                           					\\
    $\theta_i$         & Acceptance threshold ($0 \leq \theta_i < 1$) for a ridesharing route $\pi_i(\alpha_i)$ (rider only)  \\
    $\pi_i(\alpha_i)$ & Route for $i$ using a combination of public transit and ridesharing (rider only) \\
    $\hat{\pi}_i(\alpha_i)$  & Route for $i$ using only public transit (rider only) 		\\
    $d(\pi_i(\alpha_i))$     & The driver of ridesharing route $\pi_i(\alpha_i)$		   	\\
    $t(p_i)$           & Shortest travel time for traversing path $p_i$ by private vehicle	\\
    $t(\pi_i(\alpha_i))$ \& $t(\hat{\pi}_i(\alpha_i))$   & Shortest travel time for traversing route $\pi_i(\alpha_i)$ and $\hat{\pi}_i(\alpha_i)$ resp.   \\
    $t(u,v)$ \& $\hat{t}(u,v)$   & Shortest travel time from $u$ to $v$ by private vehicle and public transit resp. \\ \hline
   \end{tabular}
\caption{Parameters for a trip announcement $i$.}
\label{table-notation}
\end{table}
The maximum trip time $\gamma_i$ of a driver $i$ can be calculated as $\gamma_i = \min\{\gamma_i, t(p_i) + z_i\}$ if $p_i$ is given, where $t(p_i)$ is the shortest travel time on path $p_i$; otherwise $\gamma_i=\min\{\gamma_i, t(o_i,d_i) + z_i\}$, where $t(o_i,d_i)$ is the shortest travel time of a path from $o_i$ to $d_i$. For a rider $j$, $\gamma_j$ is more flexible; it is default to be $\gamma_j = t(\hat{\pi}_j(\alpha_j))$ in our experiment, where $\hat{\pi}_j(\alpha_j)$ is the fastest public transit route.

For a driver $i \in D$ and a set $J \subseteq R$ of riders, the set $\sigma(i) = \{i\} \cup J$ is called a {\em feasible match} if driver $i$ can serve this group $J$ of riders together while all requirements (constraints) specified by the parameters of the trips in $\sigma(i)$ are satisfied collectively, as listed below:
\begin{enumerate}
\item \textit{Ridesharing route constraint}: for $J=\{j_1,\ldots,j_p\}$, there is a path $r_i = (o_i,o_{j_1},\ldots,o_{j_p}$, $s,d_i)$ in $N$, where $s$ is the drop-off location for Type 1 match;
or there is a path $r_i = (o_i,s,d_{j_1},...,d_{j_p},d_i)$ in $N$, where $s$ is the pick-up location for Type 2 match.
Note that if $p_i$ is given and detour limit $z_i = 0$, path $r_i = p_i$ for either match type (assuming driver $i$ specifies a station $s$).
Otherwise, the centralized system computes the path $r_i$.

\item \textit{Capacity constraint}: limits the number of passengers a driver can serve, $1\leq |J| \leq n_i$ with the assumption $n_i \geq 1$.

\item \textit{Acceptable constraint}: each passenger $r_j \in J$ is given an acceptable route $\pi_j(\alpha_j)$ such that $t(\pi_j(\alpha_j)) \leq \theta_j \cdot t(\hat{\pi}_j(\alpha_j))$ for $0 < \theta_j \leq 1$, where the ridesharing part of $\pi_j(\alpha_j)$ is a subpath of $r_i$ and $\hat{\pi}_j(\alpha_j)$ is the fastest public transit route for $r_j$ given $\alpha_j$.

\item \textit{Travel time constraint}: each trip $j \in \sigma(i)$ departs from $o_j$ no earlier than $\alpha_j$, arrives at $d_j$ no later than $\beta_j$, and the total travel duration of $j$ is at most $\gamma_j$.
The exact application of these time constraints is described in Section~\ref{subsection-alg-feas-single} (Algorithm 1) and Section~\ref{subsection-alg-feas-all} (Algorithm 2).

\item \textit{Stop constraint}: the number of unique locations visited by driver $i$ to pick-up (for Type 1) or drop-off (for Type 2) all passengers of $\sigma(i)$ is at most $\delta_i$.
\end{enumerate}
We make two simplifications in our algorithms:
\begin{itemize}
\item Given an origin $o_j$ and a destination $d_j$ of a rider $j$ with earliest departure time $\alpha_j$ at $o_j$, we use a simplified transit system in our experiments to calculate the fastest public transit route $\hat{\pi}_j(\alpha_j)$ from $o_j$ to $d_j$.

\item We use a simplified model for the transit travel time, transit waiting time and ridesharing service time (time it takes to pick-up and drop-off riders, walking time between locations and stations).
Given the fastest travel time $t(u,v)$ by car from location $u$ to location $v$, we multiply a small constant $\epsilon>1$ with $t(u,v)$ to simulate the transit time and ridesharing service time.
In this model, the transit time and ridesharing service time are considered together, as a whole.
\end{itemize}

\subsubsection{Phase one (Algorithm 1)}\label{subsection-alg-feas-single}
We now describe how to compute a feasible match between a driver and a rider for Type 1. The computation for Type 2 is similar and we omit it.
For every trip $i \in D \cup R$, we first compute the set $S_{do}(i)$ of feasible drop-off locations for trip $i$.
Each element in $S_{do}(i)$ is a \emph{station-time} tuple $(s, \alpha_i(s))$ of $i$, where $\alpha_i(s)$ is the earliest possible time $i$ can reach station $s$.
The station-time tuples are computed by the following preprocessing procedure.
\begin{itemize}[leftmargin=*]
\item We find all feasible station-time tuples for each rider $j \in R$. A station $s$ is \emph{time feasible} for $j$ if $j$ can reach $d_j$ from $s$ within time window $[\alpha_j, \beta_j]$ and $t(o_j,s) + \hat{t}(s,d_j) \leq \min\{\gamma_j, \theta_j \cdot \hat{t}(o_j,d_j)\}$.
	\begin{itemize}
	\item The earliest possible time to reach station $s$ for $j$ can be computed as $\alpha_j(s) = \alpha_j + t(o_j,s)$ without pick-up time and drop-off time. Since we do not consider waiting time and ridesharing service time separately, $\alpha_j(s)$ also denotes the earliest time for $j$ to depart from station $s$.
	
	\item Let $\hat{t}(s,d_j)$ be the travel time of a fastest public route. Station $s$ is \emph{time feasible} if $\alpha_j(s) +\hat{t}(s,d_j) \leq \beta_j$ and $t(o_j,s) + \hat{t}(s,d_j) \leq \min\{\gamma_j, \theta_j \cdot \hat{t}(o_j,d_j)\}$.
	\end{itemize}
	
\item The feasible station-time tuples for each driver $i \in D$ is computed by a similar calculation.
	\begin{itemize}[leftmargin=*]
    \item Without considering pick-up time and drop-off time separately, the earliest arrival time of $i$ to reach $s$ is $\alpha_i(s) = \alpha_i + t(o_i,s)$. Station $s$ is \emph{time feasible} if $\alpha_i(s) + t(s,d_i) \leq \beta_i$ and $t(o_i,s) + t(s,d_i) \leq \gamma_i$.
	\end{itemize}
\end{itemize}

After the preprocessing, Algorithm~\ref{alg-feas-single} finds all feasible matches, each consists of a single rider.
For each pair $(i, j)$ in $D \times R$, let $\eta_i(o_j) = \max\{\alpha_i,\alpha_j - t(o_i,o_j)\}$ be the latest departure time of driver $i$ from $o_i$ such that $i$ can still pick-up $j$ at the earliest; this minimizes the time (duration) needed for driver $i$ to wait for rider $j$. Hence, the total travel time of $i$ is minimized when $i$ uses a path $r_i$ with shortest travel time and departure time $\eta_i(o_j)$.
The process of checking if the match $\sigma(i) = \{i,j\}$ is feasible for all pairs of $(i,j)$ can be performed as in Algorithm~\ref{alg-feas-single}.
\begin{algorithm}[ht]
\small
\begin{algorithmic}[1]
\For {each pair $(i, j)$ in $D \times R$}
      \For {each station $s$ in $S_{do}(i) \cap S_{do}(j)$}
            \State $t_1 = t(o_i,o_j) + t(o_j,s)$; $t_2 = t(o_j,s)$; \hspace*{2mm} // travel duration for $i$ and $j$ to reach $s$ resp.
            \State $t = \eta_i(o_j) + t_1$; \hspace*{15mm} // earliest departure time from station $s$
            \State \parbox[t]{385pt}{%
                  \textbf{if} $(t + t(s, d_i) \leq \beta_i \wedge t_1+t(s, d_i) \leq \gamma_i)$ and $(t + \hat{t}(s, d_j) \leq \beta_j \wedge t_2 + \hat{t}(s, d_j) \leq \min\{\gamma_j, \theta_j \cdot \hat{t}(o_j,d_j)\})$ \textbf{then}}
            \State\hspace*{1.5em} create an edge $(i, j)$ in $E(H)$ for $\sigma(i) = \{i,j\}$
            \State\hspace*{1.5em} \textbf{break} inner for-loop; \hspace*{2mm} // can continue the for-loop for a better route
            \State \textbf{end if}
    \EndFor
\EndFor
\end{algorithmic}
\caption{(Phase one) compute all feasible matches, each consists of a single rider}
\label{alg-feas-single}
\end{algorithm}

\subsubsection{Phase two (Algorithm 2)}\label{subsection-alg-feas-all}
We extend Algorithm~\ref{alg-feas-single} to create matches with more than one rider.
Let $H(V,E)$ be the graph after computing all feasible matches consisting of a single rider (instance computed by Algorithm~\ref{alg-feas-single}).
We start with computing, for each driver $i$, feasible matches consisting of two riders, then three riders, and so on until $\min\{\delta_i,n_i\}$.
Let $\Omega(i)$ be the set of feasible matches found so far for driver $i$ and $\Omega(i,p-1) = \{\sigma(i) \in \Omega(i) \mid |\sigma(i) \setminus \{i\}|=p-1\}$ be the set of matches with $p-1$ riders, and we try to extend $\Omega(i,p-1)$ to $\Omega(i,p)$ for $2\leq p\leq \min\{\delta_i,n_i\}$.
Let $r_i = (l_0,l_1,\ldots,l_p,s,d_i)$ denotes an ordered potential path (travel route) for driver $i$ to pick-up all $p$ riders of $\sigma(i)$ and drop-off them at station $s$, where $l_0$ is the origin of $i$ and $l_y$ is the pick-up location (origin of rider $j_y$), $1 \leq y \leq p$.
We extend the notion of $\eta_i(o_j)$, defined above in Phase one, to every pick-up location of $r_i$.
That is, $\eta_i(l_p)$ is the latest time of $\eta_i$ to depart from $o_i$ to pick-up each of the riders $j_1,\ldots,j_p$ such that the waiting time of $i$ is minimized, and hence, travel time of $i$ is minimized.
We simply call $\eta_i(l_p)$ the latest departure of $i$ to pick-up $\sigma(i)$.
All possible combinations of $r_i$ are enumerated to find a feasible path $r_i$; the process of finding $r_i$ is described in the following.
\begin{itemize}
\setlength\itemsep{0em}
\item First, we fix a combination of $r_i$ such that $|\sigma(i)| \leq n_i + 1$ and $r_i$ satisfies the stop constraint. The visiting order of the pick-up origin locations is known when we fix a path for $r_i$.
\item The algorithm determines the actual drop-off station $s$ in $r_i = (l_0,l_1,\ldots,l_{p},s,d_i)$.
Let $j_{y}$ be the rider corresponds to pick-up location $l_y$ for $1 \leq y \leq p$ and $l_0 = o_i$.
For each station $s$ in $\bigcap_{0 \leq y \leq p} S_{do}(j_y)$, the algorithm checks if $r_i = (l_0,l_1,\ldots,l_{p},s,d_i)$ admits a time feasible path for all trips in $\sigma(i)$ as follows.
    \begin{itemize}
        \item The total travel time (duration) for $i$ from $l_0$ to $s$ is $t_i = t(l_0, l_1) + \cdots + t(l_{p-1},l_{p}) + t(l_p, s)$.
        The total travel time (duration) for $j_y$ from $l_y$ to $s$ is $t_{j_y} = t(l_y,l_{y+1}) + \cdots + t(l_{p-1},l_{p}) + t(l_p, s)$, $1 \leq y \leq p$.
        \item Since the order for $i$ to pick up $j_y$ ($1 \leq y \leq p$) is fixed, $\eta_i(l_p)$ can be calculated as $\eta_i(l_p) = \max\{\alpha_i, \alpha_{j_1} - t(l_0,l_1), \alpha_{j_{2}} - t(l_0,l_1) - t(l_1,l_2), \ldots, \alpha_{j_{p}} - t(l_0,l_1) - \cdots - t(l_{p-1},l_p)\}$.
        The earliest arrival time at $s$ for all trips in $\sigma(i)$ is $t = \eta_i(l_p) + t_i$.
        \item If $t + t(s, d_i) \leq \beta_i$, $t_i + t(s, d_i) \leq \gamma_i$, and for $1\leq y\leq p$, $t + \hat{t}(s, d_{j_{y}}) \leq \beta_{j_{y}}$ and $t_{j_y} + \hat{t}(s, d_{j_{y}}) \leq \theta_{j_{y}} \cdot \hat{t}(o_{j_{y}}, d_{j_{y}})$, then $r_i$ is feasible.
    \end{itemize}
\item If $r_i$ is feasible, add the match corresponds to $r_i$ to $H$. Otherwise, check next combination of $r_i$ until a feasible path $r_i$ is found or all combinations are exhausted.
\end{itemize}
The pseudo code for the above process is given in Algorithm~\ref{alg-feas-all}.
We show that the latest departure $\eta_i(l_p)$ used in Algorithm~\ref{alg-feas-all} indeed minimizes the total travel time of $i$ to reach $l_p$.

\begin{algorithm}[!ht]
\small
\begin{algorithmic}[1]
\For {$i$ = 1 to $|D|$}
    \State $p = 2$;
    \While {($p \leq \min\{\delta_i,n_i\}$ and $\Omega(i,p-1) \neq \emptyset$)}
    \For {each match $\sigma(i)$ in $\Omega(i,p-1)$}
        \For {each $j \in R$ s.t. $j \notin \sigma(i)$}
            \State // check if $\sigma(i) \cup \{j\}$ satisfies Observation~\ref{obs-base-match}, and if not, skip $j$
            \State \textbf{if} {$((\sigma(i) \setminus \{q\}) \cup \{j\}) \in \Omega(i,p-1)$ for all $q \in \sigma(i) \setminus \{i\}$} \textbf{then}
                \State \hspace{0.5cm} \textbf{if} {($\sigma(i) \cup \{j\}$ has not been checked) and (feasibleInsert($\sigma(i), j$))} \textbf{then}
                    \State \hspace{1cm} create an edge $(i, J)$ in $E(H)$ for a feasible match $\sigma(i) = \{i\}\cup J$.
                    \State \hspace{1cm} add $\sigma(i) \cup \{j\}$ to $\Omega(i,p)$.
                \State \hspace{0.5cm} \textbf{end if}
         \EndFor
    \EndFor
    \State $p = p + 1$;
    \EndWhile
\EndFor
\\
\textbf{Procedure} feasibleInsert($\sigma(i), j$) \hspace{3mm} // find a feasible path for $i$ to serve $\sigma(i) \cup \{j\}$ if exists
\\ Let $r_i = (l_0,l_1,\ldots,l_p,s,d_i)$ denotes a potential path for driver $i$ to serve trips in $\sigma(i) \cup \{j\}$.
\For {each station $s$ in $\bigcap_{0 \leq y \leq p} S_{do}(j_y)$}
    \For {each combination of $r_i = (l_0,\ldots,l_p,s,d_i)$ that satisfies the stop constraint}
        \State $t_i = t(l_0, l_1) + \cdots + t(l_{p-1},l_{p}) + t(l_p, s)$; $t_{j_y} = t(l_y,l_{y+1}) + \cdots + t(l_{p-1},l_{p}) + t(l_p, s)$;
        \State $t = \eta_i(l_p) + t_i$; /* the earliest arrival time at $s$ for all trips in $\sigma(i)$ */
         \State \parbox[t]{385pt}{%
                  \textbf{if} $(t + t(s, d_i) \leq \beta_i \wedge t_i + t(s, d_i) \leq \gamma_i)$ and (for $1\leq y \leq p$, $t + \hat{t}(s, d_{j_{y}}) \leq \beta_{j_{y}} \wedge t_{j_y} + \hat{t}(s, d_{j_{y}}) \leq \min\{\gamma_j, \theta_j \cdot \hat{t}(o_{j_{y}}, d_{j_{y}})\})$] \textbf{then}\strut}
            \State\hspace*{1.5em} \Return True;
            \State \textbf{end if}
    \EndFor
\EndFor
\\ \Return False;
\end{algorithmic}
\caption{(Phase two) compute all feasible matches}
\label{alg-feas-all}
\end{algorithm}

\begin{theorem}
Given a feasible path $r_i = (l_0,\ldots,l_p,s,d_i)$ for driver $i$ to serve $p$ passengers in a match $\sigma(i)$.
The latest departure time $\eta_i(l_p)$ calculated above minimizes the total travel time of $i$ to reach $l_p$.
\end{theorem}

\begin{proof}
Prove by induction. For the base case $\eta_i(l_1) = \max\{\alpha_i,\alpha_{j_{1}} - t(l_0,l_{1})\}$, and by choosing departure time $\eta_i(l_1)$, driver $i$ does not need to wait for rider $j_1$ at $\alpha_{j_{1}}$. Hence, using a shortest (time) path from $l_0$ to $l_1$ with departure time $\eta_i(l_1)$ minimizes the travel time of $i$ to pick-up $j_1$.
Assume the lemma holds for $1 \leq y-1 < p$, that is, $\eta_i(l_{y-1})$ minimizes the total travel time of $i$ to reach $l_{y-1}$. We prove for $y$.
From the calculation of $\eta_i(l_{y-1})$, $\eta_i(l_{y}) = \max\{\eta_i(l_{y-1}), \alpha_{j_{y}} - t(l_0,l_{1}) - t(l_{1},l_{2}) - \cdots - t(l_{y-1},l_y)\}$. By the induction hypothesis, $\eta_i(l_{y})$ minimizes the total travel time of $i$ when using a shortest path $(l_0,\ldots,l_y)$.
\end{proof}

The running time of Algorithm~\ref{alg-feas-all} heavily depends on the number of subsets of riders to be checked for feasibility.
One way to speed up Algorithm~\ref{alg-feas-all} is to use dynamic programming (or memoization) to avoid redundant checks on a same subset.
For each feasible match $|\sigma(i)|$ of $p-1$ passengers for a driver $i \in D$, we store every feasible path $r_i = (l_0,l_1,\ldots,l_{p-1},s,d_i)$ and extend from $r_i$ to insert a new trip to minimize the number of ordered potential paths we need to test.
We can further make sure that no path $r_i$ is tested twice during execution.
First, the set $R$ of riders is given a fixed ordering (based on the integer labels).
For a feasible path $r_i$ of a driver $i$, the check of inserting a new rider $j$ into $r_i$ is performed only if $j$ is larger than every rider in $r_i$ according to the fixed ordering.
Furthermore, A heuristic approach to speed up Algorithm~\ref{alg-feas-all} is given at the end of Section~\ref{sec-instances}.

\section{Approximation algorithms} \label{sec-approximate}
We show that the MTR problem is NP-hard and give approximation algorithms for the problem.
When every edge in $H(V,E)$ consists of only two vertices (one driver and one rider), the ILP~(\ref{obj-1})-(\ref{constraint-2}) formulation is equivalent to the maximum weight matching problem, which can be solved in polynomial time.
However, if the edges contain more than two vertices, they become hyperedges. In this case, the ILP~(\ref{obj-1})-(\ref{constraint-2}) becomes a formulation of the maximum weighted set packing problem (MWSP), which is NP-hard in general~\citep{Garey79,Karp72}.
In fact, the ILP~(\ref{obj-1})-(\ref{constraint-2}) formulation gives a special case of MWSP (due to the structure of $H(V,E)$).
We first show that this special case is also NP-hard, and by Theorem~\ref{theorem-ILP}, the MTR problem is NP-hard.

\subsection{NP-hardness}
It was mentioned in~\citet{Santi-PNAS14} that their minimization problem related to shareability hyper-network is NP-complete, which is similar to the MTR problem formulation. However, an actual reduction proof was not described.
We prove the MTR problem is NP-hard by a reduction from a special case of the maximum 3-dimensional matching problem (3DM).
An instance of 3DM consists of three disjoint finite sets $A$, $B$ and $C$, and a collection $\mathcal{F} \subseteq A \times B \times C$.
That is, $\mathcal{F}$ is a collection of triplets $(a,b,c)$, where $a \in A, b \in B$ and $c \in C$.
A 3-dimensional matching is a subset $\mathcal{M} \subseteq \mathcal{F}$ such that all sets in $\mathcal{M}$ are pairwise disjoint.
The decision problem of 3DM is that given $(A, B, C, \mathcal{F})$ and an integer $q$, decide whether there exists a matching $\mathcal{M} \subseteq \mathcal{F}$ with $|\mathcal{M}| \geq q$.
We consider a special case of 3DM: $|A| = |B| = |C| = q$; it is still NP-complete~\citep{Garey79,Karp72}.
Given an instance $(A,B,C,\mathcal{F})$ of 3DM with $|A| = |B| = |C| = q$, we construct an instance $H(V,E)$ (bipartite hypergraph) of the MTR problem as follows:
\begin{itemize}
\setlength\itemsep{0em}
\item Create a set of drivers $D(H) = A$ with capacity $n_i=2$ for every driver $i\in D(H)$ and a set of riders $R(H) = B \cup C$.

\item For each $f \in \mathcal{F}$, create a hyperedge $e(f)$ in $E(H)$ containing elements $(a,b,c)$, where $a$ represents a driver and $b, c$ represent two different riders.
Further, create edges $e'(f) = (a, b)$ and $e''(f) = (a, c)$ so that Observation~\ref{obs-base-match} is satisfied.
\end{itemize}

\begin{theorem} \label{theorem-nphard}
The MTR problem is NP-hard.
\end{theorem}

\begin{proof}
By Theorem~\ref{theorem-ILP}, we only need to prove the ILP~(\ref{obj-1})-(\ref{constraint-2}) is NP-hard, which is done by showing that an instance $(A,B,C,\mathcal{F})$ of the maximum 3-dimensional matching problem has a solution $\mathcal{M}$ of cardinality $q$ if and only if the objective function value of ILP~(\ref{obj-1})-(\ref{constraint-2}) is $2q$.

Assume that $(A,B,C,\mathcal{F})$ has a solution $\mathcal{M} = \{f_1, f_2,\ldots, f_q\}$. For each $f_i$ ($1 \leq i \leq q$), set the corresponding binary variable $x_{e(f_i)} = 1$ in ILP~(\ref{obj-1})-(\ref{constraint-2}).
Since $f_i \cap f_j = \emptyset$ for $1 \leq i \neq j \leq q$, constraint~(\ref{constraint-1}) of the ILP is satisfied.
Further, each edge $e(f_i)$ corresponding to $f_i\in \mathcal{M}$ has weight $w(e(f_i))=2$, implying the objective function value of ILP~(\ref{obj-1})-(\ref{constraint-2}) is $2q$.

Assume that the objective function value of ILP~(\ref{obj-1})-(\ref{constraint-2}) is $2q$.
Let $X=\{e(f) \in E(H) \mid x_{e(f)} = 1\}$, where $x_{e(f)}$'s are the binary variables of ILP~(\ref{obj-1})-(\ref{constraint-2}).
For every edge $e(f) \in X$, add the corresponding set $f \in \mathcal{F}$ to $\mathcal{M}$.
From constraint~(\ref{constraint-1}) of the ILP, $X$ is pairwise disjoint and $|X| \leq |D(H)|$.
Hence, $\mathcal{M}$ is a valid solution for $(A,B,C,\mathcal{F})$ with $|\mathcal{M}| = |X|$.
Since every $e(f) \in X$ contains at most two different riders and $|X| \leq |D(H)| = q$, $|X| = q$ for the objective function value to be $2q$.
Thus, $|\mathcal{M}| = q$.

The size of $H(V,E)$ is polynomial in $q$. It takes a polynomial time to convert a solution of $H(V,E)$ to a solution of the 3DM instance $(A,B,C,\mathcal{F})$ and vice versa. 
\end{proof}

\subsection{Proposed approximation algorithms}
Since solving the ILP~(\ref{obj-1})-(\ref{constraint-2}) formulation exactly is NP-hard, it may require exponential time in a worst case, which is not acceptable in practice.
One way to solve this is to have a time limit on any solver (or exact algorithm). When the time limit is reached, output the current solution or the best solution found so far.
However, this does not guarantee the quality of the solution.
Hence, it is important to use an approximation algorithm as a fallback plan.

The approximation ratio of a $\rho$-approximation algorithm for a maximization problem is defined as $\frac{w(\mathcal{M})}{w(\OPT)} \geq \rho$ for $\rho < 1$, where $w(\mathcal{M})$ and $w(\OPT)$ are the values of approximation and optimal solutions, respectively.
In this section, we give a $(1-\frac{1}{e})$-approximation algorithm and a $\frac{1}{2}$-approximation algorithm for the MTR problem.
Our $(1-\frac{1}{e})$-approximation algorithm (refer to as \textit{\textbf{LPR}}) is a simplified version of the LP-rounding based algorithm obtained by~\citet{Fleischer-SODA06}.
Our $\frac{1}{2}$-approximation algorithm (refer to as \textit{\textbf{ImpGreedy}}) is a simplified version of the simple greedy~\citep{Berman-SWAT00,Chandra-JoA01} discussed in Section~\ref{sec-app-algs}. By computing a solution directly from $H(V,E)$ without solving the independent set/weighted set packing problem, the running time and memory usage of ImpGreedy are significantly improved over the simple greedy.

\subsubsection{Description of the LPR algorithm}
The ILP~(\ref{obj-1})-(\ref{constraint-2}) formulation is a special case of the Separable Assignment Problem (SAP): 
given a set $U$ of bins, a set $I$ of items, a value $f_{ij}$ for assigning item $j$ to bin $i$, and a collection $\mathcal{I}_i$ of subsets of $I$ for each bin $i$, 
SAP asks to find an assignment of items to bins such that each bin $i$ can be assigned at most one set of $\mathcal{I}_i$, each item can be assigned to at most one bin and the total value $f_{ij}$ of the assigned item is maximized.
When only one bin $i$ is considered, the problem is called the single-bin subproblem of SAP.
It can be seen that the ILP~(\ref{obj-1})-(\ref{constraint-2}) formulation of a hypergraph $H(V,E)$ is a special case of SAP, where the bins are drivers, items are riders and the edges of $H$ are $\cup_{i\in U}\mathcal{I}_i$ with unit value $f_{ij}$ for all drivers $i$ and riders $j$.

Given a $\beta$-approximation algorithm for the single-bin subproblem of SAP,~\citet{Fleischer-SODA06} obtained a local-search $(\frac{\beta}{\beta+1}-\epsilon)$-approximation algorithm ($\epsilon>0)$ and an LP-rounding based $((1-\frac{1}{e})\beta)$-approximation algorithm for SAP.
Both of these algorithms approximate the ILP~(\ref{obj-1})-(\ref{constraint-2}).
The local-search $(\frac{\beta}{\beta+1}-\epsilon)$-approximation algorithm presented by ~\citet{Fleischer-SODA06} is not efficient if one wants to have an approximation ratio as close to 1/2 as possible, assuming $\beta\approx 1$.
This is because the number of iterations of the local-search algorithm is inverse-related to $\epsilon$.
An LP for SAP is given in~\citet{Fleischer-SODA06}, but it can have exponential number of variables due to $|\mathcal{I}_i|$ can be exponentially large in general.
By the assumption that the maximum capacity $K$ of all vehicles is a small constant, $|\mathcal{I}_i|$ is polynomially bounded in our case.
From this and unit value, the single-bin subproblem of the MTR problem can be solved efficiently ($\beta=1$).
This gives a $(1-\frac{1}{e})$-approximation algorithm for the MTR problem.
More importantly, the LP of ILP~(\ref{obj-1})-(\ref{constraint-2}) can be solved directly because $|E(H)|$ ($|\mathcal{I}_i|$) is polynomially bounded.
For completeness, we describe the LPR algorithm using our notation as follows.

\begin{enumerate}
\item Obtain a linear programming LP of ILP~(\ref{obj-1})-(\ref{constraint-2}) by relaxing the 0-1 variables $x_e$ to nonnegative real variables; and solve the LP.

\item Independently for each driver $i \in D(H)$, assign $i$ a match $\sigma(i)=\{i\}\cup J$ corresponding to the edge $e=(i,J)$ with probability $x_e$ (based on all edges containing $i$, namely, for all $e \in E_i$ such that $x_e>0$).
Let $M$ be the resulting intermediate solution, which contains a set of feasible matches.

\item For any rider $j \in R(H)$, let $M_j = \{\sigma(i)\in M \mid j \in \sigma(i)\}$ be the set of matches in $M$ containing $j$.
If $|M_j|\geq 2$, then remove $r_j$ from every match of $M_j$ except one match (any one match) of $M_j$.
Finally, remove from $M$ every match $\sigma(i)=\{\eta_i\}$.
\end{enumerate}
The matches in $M$ are pairwise disjoint.
From Step 2, no two matches of $M$ contain a same driver.
From Step 3, no two matches of $M$ contain a same rider.
In Step 3, after the removal of a set of riders $J$ from a match $\sigma(i) \in M$, $\sigma(i)\setminus J$ is still a feasible match if $|\sigma(i)\setminus J|\geq 2$ by Observation~\ref{obs-base-match}.
Therefore, $M$ is a feasible solution to an instance $(N,\mathcal{A},T)$ of the MTR problem.

\begin{theorem}\label{theorem-RLP}
Let $\OPT$ be the objective function value of the ILP~(\ref{obj-1})-(\ref{constraint-2}) formulation, which is the maximum number of riders can be served.
Then the expected value $Q$ of the rounded solution $M$ of Algorithm LPR is at least $(1-\frac{1}{e})\OPT$.
\end{theorem}

\begin{proof}
Let $\OPT^*$ be the objective function value of the LP relaxation. Then $\OPT^*\geq \OPT$.
From Theorem 2.1 in~\citet{Fleischer-SODA06}, $Q\geq (1-(1-\frac{1}{m})^m)\OPT^*\geq (1-\frac{1}{e}) \OPT$, where $m=|M|$.
Since $M$ is a feasible solution as explained above, the theorem holds.
\end{proof}

\subsubsection{Description of the ImpGreedy algorithm}
Algorithm ImpGreedy is similar to the $\frac{1}{K+1}$-approximation algorithm obtained by~\citet{Santi-PNAS14} assuming the maximum capacity $K$ among all vehicles (drivers) is at least two.
However, a detailed analysis for the approximation ratio of their greedy algorithm is not presented in~\citet{Santi-PNAS14}.
Hence, in this section, we describe our ImpGreedy algorithm along with a complete proof for its constant $\frac{1}{2}$-approximation ratio.
For the hypergraph $H(V,E)$ constructed for an instance $(N,\mathcal{A},T)$ of the MTR problem, 
denoted by $\Sigma \subseteq E(H)$ is the current partial solution computed by ImpGreedy (recall that each edge of $E(H)$ represents a feasible match).
Let $P(\Sigma)=\bigcup_{e \in \Sigma} R(e)$, called the \textit{covered passengers}.
Initially, $\Sigma = \emptyset$.
In each iteration, we add an edge with the most number of uncovered passengers to $\Sigma$, that is, select an edge $e$ such that $|R(e)|$ is maximum, and then add $e$ to $\Sigma$.
Remove $E_e = \cup_{j \in e} E_j$ from $E(H)$ ($E_j$ is defined in Section~\ref{sec-exact-IP}).
Repeat until $|P(\Sigma)| = |R(H)|$ or $|\Sigma| = |D(H)|$.
The pseudo code of ImpGreedy is shown in Algorithm~\ref{alg-new-approx}.

\begin{algorithm}
\small
\begin{algorithmic}[1]
\State \textbf{Input:} The hypergraph $H(V,E)$ for problem instance $(N,\mathcal{A},T)$.
\State \textbf{Output:} A solution $\Sigma$ to $(N,\mathcal{A},T)$ with $\frac{1}{2}$-approximation ratio.
\State $\Sigma = \emptyset$; $P(\Sigma) = \emptyset$;
\While{($|P(\Sigma)| < |R(H)|$ and $|\Sigma| < |D(H)|$)}
    \State compute $e = \text{argmax}_{e \in E(H)} |R(e)|$;
    $\Sigma = \Sigma \cup \{e\}$; update $P(\Sigma)$; remove $E_e$ from $E(H)$;
\EndWhile
\end{algorithmic}
\caption{ImpGreedy}
\label{alg-new-approx}
\end{algorithm}

\subsubsection{Analysis of ImpGreedy}
In ImpGreedy, when an edge $e$ is added to $\Sigma$, $E_e$ is removed from $E(H)$, so Property~\ref{property-gamma} holds for $\Sigma$.
Further, the edges in $\Sigma$ are pairwise vertex-disjoint, implying $\Sigma$ is a feasible solution.
\begin{property}
For every $i \in D(H)$, at most one edge $e$ from $E_i$ can be selected in any solution.
\label{property-gamma}
\end{property}
Let $\Sigma = \{x_1, x_2,\ldots, x_a\}$ be a solution found by Algorithm ImpGreedy, where $x_i$ is the $i^{th}$ edge added to $\Sigma$.
Throughout the analysis, we use $\OPT$ to denote an optimal solution, that is, $\OPT$ is a set of edges that are pairwise vertex-disjoint and $R(\OPT) \geq R(\Sigma)$.
Further, $\Sigma_i = \bigcup_{1 \leq b \leq i} x_b$ for $1 \leq i \leq a$, $\Sigma_0 = \emptyset$ with $R(\Sigma_0) = \emptyset$, and $\Sigma_a = \Sigma$.
Since each edge $e$ of $E(H)$ represents a feasible match, we overload any edge $x_i \in \Sigma$ to denote a match as well.
For each $x_i\in \Sigma$, by Property~\ref{property-gamma}, there is at most one $y \in \OPT$ with $D(y)=D(x_i)$.
We order $\OPT$ and introduce dummy edges to $\OPT$ such that $D(y_i) = D(x_i)$ for $1 \leq i\leq a$.
Formally, for $1\leq i\leq a$, define
\begin{small}
\[
\OPT(i)=\{y_1,\ldots,y_i \mid 1\leq b \leq i, D(y_b)=D(x_b) \text{ if } y_b \in \OPT, \text{ otherwise } y_b \text{ is a dummy edge}\}.
\]
\end{small}%
A dummy edge $y_b\in \OPT(i)$ is defined as $D(y_b) = D(x_b)$ with $R(y_b)=\emptyset$.
Notice that there can be edges $y$ in $OPT$ such that $D(y) \neq D(x)$ for every $x \in \Sigma$.
Such edges are in $\OPT\setminus \OPT(a)$.

\begin{lemma}
Let $\OPT$ be an optimal solution and $\Sigma = \{x_1, \ldots, x_a\}$ be a solution found by ImpGreedy.
For any $1 \leq i \leq a$, $|R(y_i) \setminus R(\Sigma_{i-1})| \leq |R(x_i)|$.
\label{lemma-max-gap}
\end{lemma}

\begin{proof}
If $|R(y_i)| \leq |R(x_i)|$, then the lemma holds.
Suppose $|R(y_i)| > |R(x_i)|$.
Since the algorithm selects $x_i$ instead of $y_i$, it must mean that $R(y_i) \cap R(\Sigma_{i-1}) \neq \emptyset$, and $y_i$ has been removed from $E(H)$ while searching for $x_i$.
By Observation~\ref{obs-base-match}, there is an edge $e_{z} \in E(H)$ such that $D(e_z) = D(y_i)$ and $R(e_{z}) = R(y_i) \setminus R(\Sigma_{i-1})$; and $|R(e_z)|\leq |R(x_i)|$ from the algorithm.
Hence, $|R(y_i) \setminus R(\Sigma_{i-1})| \leq |R(x_i)|$.
\end{proof}

\begin{lemma}
Let $\OPT' = \OPT\setminus \OPT(a)$. Then, $R(\OPT') \subseteq R(\Sigma)$.
\label{lemma-opt-subset}
\end{lemma}

\begin{proof}
Assume for contradiction that there exists an edge $y \in \OPT'$ s.t. $R(y) \setminus R(\Sigma) \neq \emptyset$.
By Observation~\ref{obs-base-match}, there is an edge $e_z \in E(H)$ such that $D(e_{z}) = D(y)$ and $R(e_z) = R(y)\setminus R(\Sigma)$, and $e_z \notin \Sigma$.
Since $e_z$ is not incident to any vertex of $D(\Sigma) \cup R(\Sigma)$, the algorithm should have added $e_z$ to $\Sigma$, a contradiction.
\end{proof}

\begin{theorem}
Given a hypergraph instance $H(V,E)$, Algorithm ImpGreedy computes a solution $\Sigma$ for $H$ such that $\frac{|R(\Sigma)|}{|R(\OPT)|} \geq \frac{1}{2}$, where $\OPT$ is an optimal solution, with running time $O(|D(H)| \cdot |E(H)|)$ and $|E(H)| = O(|D| \cdot (|R|+1)^K)$.
\label{theorem-ImpGreedy}
\end{theorem}

\begin{proof}
Let $\Sigma = \{x_1,\ldots,x_a\}$, $\OPT(a)$ as defined above, and $\OPT' = \OPT\setminus \OPT(a)$.
From Lemma~\ref{lemma-max-gap}, we get
\[
|\bigcup^a_{i=1} R(y_i) \setminus R(\Sigma_{i-1})| \leq |\bigcup^a_{i=1} R(x_i)| = |R(\Sigma)|.
\]
From this and Lemma~\ref{lemma-opt-subset}, we obtain

\begin{align*}
|R(\OPT)| &= |R(\OPT(a)) \cup R(\OPT')| \\
&\leq |(\bigcup^a_{i=1} (R(y_i) \setminus R(\Sigma_{i-1})) \cup R(\Sigma)) \cup R(\OPT')| \\
&= |(\bigcup^a_{i=1} R(y_i) \setminus R(\Sigma_{i-1}))| + |R(\Sigma)| \\
&\leq |R(\Sigma)| + |R(\Sigma)| = 2|R(\Sigma)|.
\end{align*}

In each iteration of the while-loop, it takes $O(|E(H)|)$ to find an edge $x$ with maximum $|R(x)|$,
and there are at most $|D(H)|$ iterations. Hence, Algorithm ImpGreedy runs in $O(|D(H)| \cdot |E(H)|)$ time.
\end{proof}

\subsection{Approximation algorithms for maximum weighted set packing}\label{sec-app-algs}
Now, we explain the algorithms for the maximum weighted set packing problem, which can also solve the MTR problem.
Given a universe $\mathcal{U}$ and a family $\mathcal{S}$ of subsets of $\mathcal{U}$, a \emph{packing} is a subfamily $\mathcal{C} \subseteq \mathcal{S}$ of sets such that all sets in $\mathcal{C}$ are pairwise disjoint.
In the maximum weighted $k$-set packing problem (MWSP), every subset $S \in \mathcal{S}$ has at most $k$ elements and is given a real weight, and MWSP asks to find a packing $\mathcal{C}$ with the largest total weight.
We can see that the ILP~(\ref{obj-1})-(\ref{constraint-2}) formulation of a hypergraph $H(V,E)$ is a special case of the maximum weighted $k$-set packing problem, where the trips of 
$D(H) \cup R(H)$ is the universe $\mathcal{U}$ and $E(H)$ is the family $\mathcal{S}$ of subsets, and every $e \in E(H)$ is a set in $\mathcal{S}$ representing at most $k = K+1$ trips ($K$ is the maximum capacity of all vehicles).
Hence, solving MWSP also solves the MTR problem.
\citet{Hazan-CC06} showed that the $k$-set packing problem cannot be approximated to within $O(\frac{\text{ln} k}{k})$ in general unless P = NP.
\citet{Chandra-JoA01} presented a $\frac{3}{2(k+1)}$-approximation and a $\frac{5}{2(2k+1)}$-approximation algorithms (refer to as \textit{\textbf{BestImp}} and \textit{\textbf{AnyImp}}, respectively), and
~\citet{Berman-SWAT00} presented a $(\frac{2}{k+1})$-approximation algorithm (referred to as \textit{\textbf{SquareImp}}) for the weighted $k$-set packing problem. Here, $k \geq 3$.

The three algorithms (AnyImp, BestImp and SquareImp) in~\citet{Berman-SWAT00} and \citet{Chandra-JoA01}  solve the weighted $k$-set packing problem by first transferring it into a weighted independent set problem, which consists of a vertex weighted graph $G(V,E)$ and asks to find a maximum weighted independent set in $G(V,E)$.
We briefly describe the common local search approach used in these three approximation algorithms.
A \emph{claw} $C$ in $G$ is defined as an induced connected subgraph that consists of an independent set $T_C$ of vertices (called talons) and a center vertex $C_z$ that is connected to all the talons ($C$ is an induced star with center $C_z$).
For any vertex $v \in V(G)$, let $N(v)$ denotes the set of vertices in $G$ adjacent to $v$, called the \emph{neighborhood} of $v$.
For a set $U$ of vertices, $N(U) = \cup_{v \in U} N(v)$.
The \textit{local search} of AnyImp, BestImp and SquareImp uses the same central idea, summarized as follows:
\begin{enumerate}
\setlength\itemsep{0em}
\item The approximation algorithms start with an initial solution (independent set) $I$ in $G$ found by a \textit{simple greedy} (referred to as \textit{\textbf{Greedy}}) as follows: select a vertex $u \in V(G)$ with largest weight and add to $I$.
Eliminate $u$ and all $u$'s neighbors from being selected. Repeatedly select the largest weight vertex until all vertices are eliminated from $G$.
\item While there exists claw $C$ in $G$ w.r.t. $I$ such that independent set $T_C$ improves the weight of $I$ (different for each algorithm),
augment $I$ as $I = (I \setminus N(T_C)) \cup T_C$; such an independent set $T_C$ is called an \emph{improvement}.
\end{enumerate}
To apply these algorithms to the MTR problem, we need to convert the bipartite hypergraph $H(V,E)$ to a weighted independent set instance $G(V,E)$, which is straightforward.
Each hyperedge $e \in E(H)$ is represented by a vertex $v_e \in V(G)$. The weight $w(v_e) = w(e)$ for each $e \in E(H)$ and $v_e \in V(G)$. There is an edge between $v_{e}, v_{e'} \in V(G)$ if $e \cap e' \neq \emptyset$ where $e, e' \in E(H)$.
We observed the following property.

\begin{property}
When the size of each set in the set packing problem is at most $k$ $(|w(e)| = k-1, e \in E(H))$, the graph $G(V,E)$ has the property that it is $(k+1)$-claw free, that is, $G(V,E)$ does not contain an independent set of size $k+1$ in the neighborhood of any vertex.
\end{property}

Applying this property, we only need to search a claw $C$ consists of at most $k$ talons, which upper bounds the running time for finding a claw within $O(n^k)$, where $n = |V(G)|$.
When $k$ is very small, it is practical enough to approximate the ILP~(\ref{obj-1})-(\ref{constraint-2}) formulation of a hypergraph $H(V,E)$ computed by Algorithm~\ref{alg-feas-all}.
It has been mentioned in~\citet{Santi-PNAS14} that the approximation algorithms in~\citet{Chandra-JoA01} can be applied to their ridesharing problem. However, only the simple greedy (\textit{Greedy}) was implemented in~\citet{Santi-PNAS14}.
Notice that ImpGreedy (Algorithm 3) is a simplified version of the Greedy algorithm, and Greedy is used to get an initial solution in algorithms AnyImp, BestImp and SquareImp. From Theorem~\ref{theorem-ImpGreedy}, we have Corollary~\ref{corollary-approximate}.
\begin{corollary}
Each of Greedy, AnyImp, BestImp and SquareImp algorithms computes a solution to $H(V,E)$ with $\frac{1}{2}$-approximation ratio.
\label{corollary-approximate}
\end{corollary}
Since ImpGreedy finds a solution directly on $H(V,E)$ without converting it to an independent set problem $G(V,E)$ and solving it, ImpGreedy is more time and space efficient than the algorithms for MWSP.
In the rest of this paper, Algorithm 3 is referred to as ImpGreedy.

\section{Numerical experiments} \label{sec-experiment}
We create a simulation environment consisting of a centralized system that integrates public transit and ridesharing.
The centralized system receives batches of discrete driver and rider trips continuously.
We implement the approximation algorithms ImpGreedy, LPR, Greedy, AnyImp and BestImp, and an exact algorithm that solves ILP formulation \eqref{obj-1}-\eqref{constraint-2} to evaluate the benefits of having such an integrated transportation system.
The results of SquareImp are not discussed because its performance is the same as AnyImp when using the smallest improvement factor ($\alpha > 1$ in~\citet{Chandra-JoA01}); this is due to the implementation of the independent set instance $G(V,E)$ having a fixed search/enumeration order of the vertices and edges, and each vertex in $V(G)$ has an integer weight.

We use a simplified transit network of Chicago to simulate the public transit and ridesharing.
The data instances generated in our experiments focus more on trips that commute to and from work (to and from the downtown area of Chicago).
To the best of our knowledge, a mass transportation system in large cities integrating public transit and ridesharing has not been implemented in real-life.
There is not any large dataset containing customers that use both public transit and ridesharing transportation modes together.
Hence, we use two related datasets to generate representative instances for our experiments.
One dataset contains transit ridership data, and the other dataset contains ridesharing trips data.
The transit ridership dataset allows us to determine the busiest transit routes, and we use this information to create rider demand in these busiest regions.
We assume riders of longer transit trips would like to reduce their travel duration by using the integrated ridesharing service.
The ridesharing dataset reveals whether there are enough personal drivers willing to provide ridesharing services.
We understand that these drivers may not be the ones who driver their vehicles to work, but at least it shows that there are currently enough drivers to support the proposed transportation system.

\subsection{Description and characteristics of the datasets}
We built a simplified transit network of Chicago to simulate practical scenarios of public transit and ridesharing.
The roadmap data of Chicago is retrieved from OpenStreetMap\footnote{Planet OSM. \url{https://planet.osm.org}}.
We used the GraphHopper\footnote{GraphHopper 1.0. \url{https://www.graphhopper.com}} library to construct the logical graph data structure of the roadmap, which contains 177037 vertices and 263881 edges.
The Chicago city is divided into 77 official community areas, each of which is assigned an area code.
We examined two different datasets in Chicago to reveal some basic traffic pattern (the datasets are provided by the Chicago Data Portal (CDP) and Chicago Transit Authority (CTA)\footnote{CDP. \url{https://data.cityofchicago.org}. CTA. \url{https://www.transitchicago.com}}, maintained by the City of Chicago).
The first dataset contains bus and rail ridership, which shows the monthly averages and monthly totals for all CTA bus routes and train station entries. We denote this dataset as \textit{PTR, public transit ridership}.
The PTR dataset range is chosen from June 1st, 2019 to June 30th, 2019.
The second dataset contains rideshare trips reported by Transportation Network Providers (sometimes called rideshare companies) to the City of Chicago. We denote this dataset as \textit{TNP}.
The TNP dataset range is chosen from June 3rd, 2019 to June 30th, 2019, total of 4 weeks of data.
Table~\ref{table-PTRdata} and Table~\ref{table-TNPdata} show some basic stats of both datasets.
\begin{table}[htbp]
\small
\captionsetup{font=small}
\parbox{.5\linewidth}{
\centering
\begin{tabular}{| p{3.7cm} | p{3.3cm} |}
\hline
Total Bus Ridership & 20,300,416  \\
Total Rail Ridership & 19,282,992 \\ \hline
12 busiest bus routes & 3, 4, 8, 9, 22, 49, 53, 66, 77, 79, 82, 151 \\ \hline
The busiest bus routes selected & 4, 9, 49, 53, 77, 79, 82 \\
\hline
\end{tabular}
\caption{Basic stats of the PTR dataset. \label{table-PTRdata}}
}
\hfill
\parbox{.49\linewidth}{
\centering
\begin{tabular}{| p{4.3cm} | p{2.8cm} |}
\hline
\# of original records & 8,820,037 \\ \hline 
\# of records considered & 7,427,716 \\ 
\# of shared trips & 1,015,329 \\ 
\# of non-shared trips & 6,412,387 \\ \hline
The most visited community areas selected & 1, 4, 5, 7, 22, 23, 25, 32, 41, 64, 76 \\
\hline
\end{tabular}
\caption{Basic stats of the TNP dataset. \label{table-TNPdata}}
}
\end{table}

In the PTR dataset, the total ridership for each bus route is recorded; there are 127 bus routes in the dataset.
We examined the 12 busiest bus routes based on the total ridership. 
7 out of the 12 routes are selected (excluding bus routes that are too close to train stations) as listed in Table~\ref{table-PTRdata} to support the selection of the community areas.
We also selected all the major trains/metro lines within the Chicago area except the Brown Line and Purple Line since they are too close to the Red and Blue lines.
Note that the PTR dataset also provides the total rail ridership. However, it only provides the number of riders entering every station in each day; it does not provide the number of riders exiting a station nor the time related to the entries.

Each record in the TNP dataset describes a passenger trip served by a driver who provides the rideshare service;
a trip record consists of a pick-up and a drop-off time and a pick-up and a drop-off community area of the trip, and exact locations are not provided.
We removed records where the pick-up or drop-off community area is hidden for privacy reason or not within Chicago, which results in 7.4 million ridesharing trips.
We calculated the average number of trips per day departed from and arrived at each area.
The results are plotted in Figure~\ref{fig-OD-pairs}; the community areas that have the highest numbers of departure trips are almost the same as that of the arrival trips.
\begin{figure}[!ht]
\centering
\includegraphics[width=\textwidth]{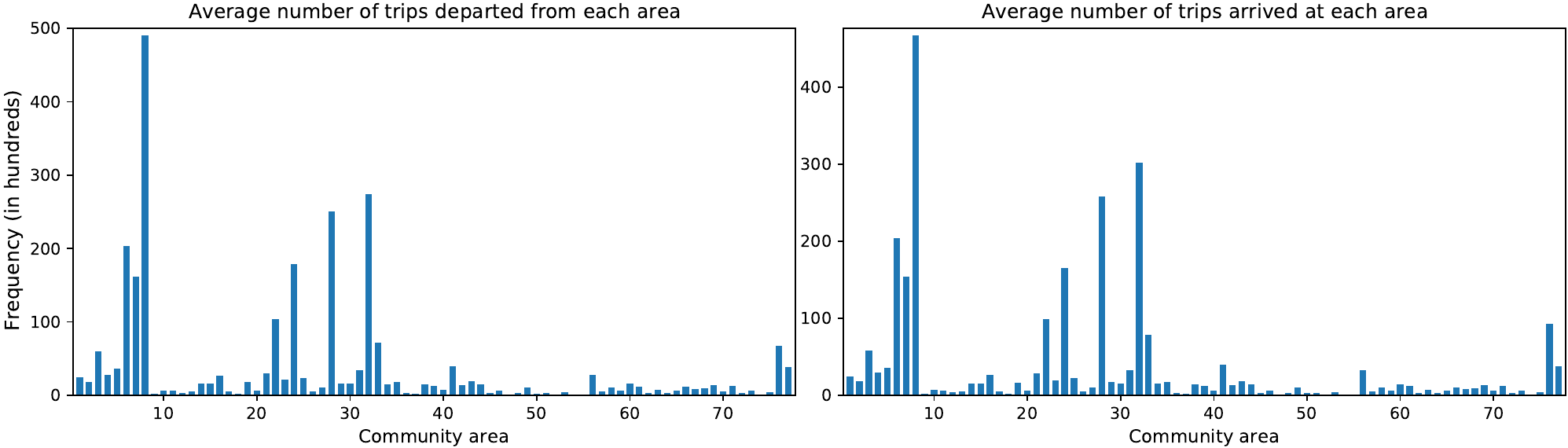}
\captionsetup{font=small}
\caption{The average number of trips per day departed from and arrived at each area.}
\label{fig-OD-pairs}
\end{figure}

We selected 11 of the 20 most visited areas as listed in Table~\ref{table-TNPdata} (area 32 is Chicago downtown, areas 64 and 76 are airports) to build the transit network for our simulation.
From the selected bus routes, trains and community areas (22 areas in total), we created a simplified public transit network connecting the community areas, depicted in Figure~\ref{fig-transit-network}.
\begin{figure}[!ht]
\centering
\includegraphics[width=\textwidth]{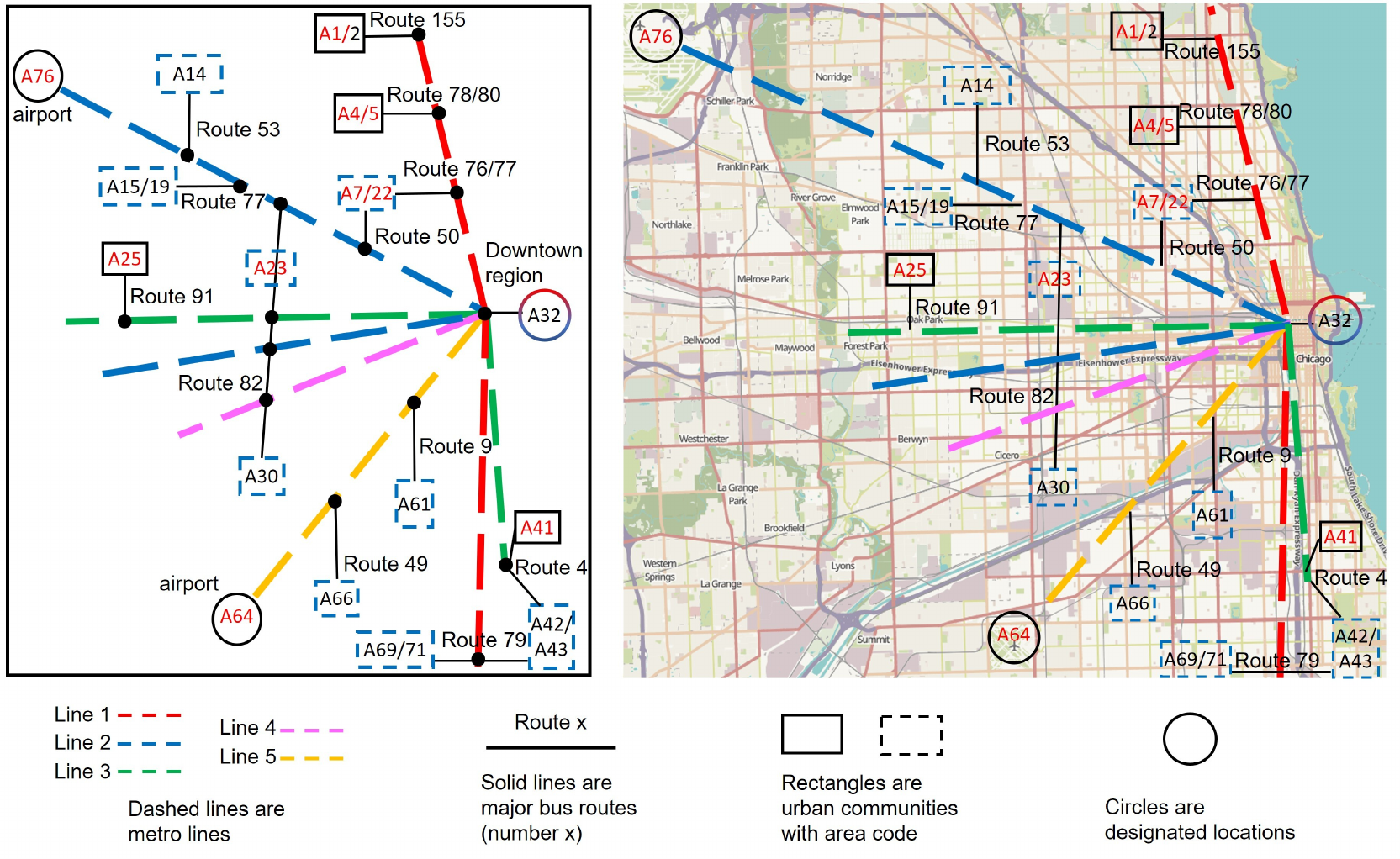}
\captionsetup{font=small}
\caption{Simplified public transit network of Chicago with 19 urban community areas and 3 designated locations (minor bus routes are not shown). Figure on the right has the Chicago City map overlay for scale.}
\label{fig-transit-network}
\end{figure}
Three of the 22 community areas are the \textit{designated locations} which include the downtown region in Chicago and the two airports. We label the rest of the 19 community areas as \textit{urban community areas}.
Each rectangle on the figure represents an \textit{urban community} within one urban community area or across two urban community areas, labeled in the rectangle.
The blue dashed rectangles/urban communities are chosen due to the busiest bus routes from the PTR dataset.
The rectangles/urban communities labeled with red area codes are chosen due to the most visited community areas from the TNP dataset.
The dashed lines are the trains, which resemble the major train services in Chicago. The solid lines are the selected bus routes connecting the urban communities to their closest train stations.
We assume that there is a major bus route travels within each urban community or some minor bus route (not labeled in Figure~\ref{fig-transit-network}) that travels to the nearest train station from each urban community.
From the datasets, many people travel to/from the designated locations (downtown region and the two airports).

The travel time between two locations by car (each location consists of the latitude and longitude coordinates) uses the fastest/shortest route computed by the GraphHopper library.
The shortest paths are \textbf{computed in real-time}, unlike many previous simulations where the shortest paths are pre-computed and stored.
As mentioned in Section~\ref{subsection-alg-feas-single}, transit travel and waiting time (transit time for short) and service time are considered in a simplified model; we multiply a small constant $\epsilon > 1$ to the fastest route to mimic transit time and service time.
For instance, consider two consecutive metro stations $s_1$ and $s_2$. The travel time $t(s_1,s_2)$ is computed by the fastest route traveled by personal cars, and the travel time by train between from $s_1$ to $s_2$ is $\hat{t}(s_1,s_2) = 1.15 \cdot t(s_1,s_2)$. The constant $\epsilon$ for bus service is 2.
Rider trips originated from all locations (except airports) must take a bus to reach a metro station when ridesharing service is not involved.

\subsection{Generating instances}\label{sec-instances}
In our simulation, we partition a day from 6:00 to 23:59 into 72 time intervals (each has 15 minutes), and we only focus on weekdays.
To observe the common ridesharing traffic pattern, we calculated the average number of served passenger trips per hour for each day of the week using the TNP dataset.
The dashed (orange) line and solid (blue) line of the plot in Figure~(\ref{fig-sub-originalTrips}) represent shared trips and non-shared trips, respectively.
\begin{figure}[!ht]
\captionsetup{font=small}
\centering
\begin{subfigure}{.62\textwidth}
  \centering
  \includegraphics[width=1.0\linewidth]{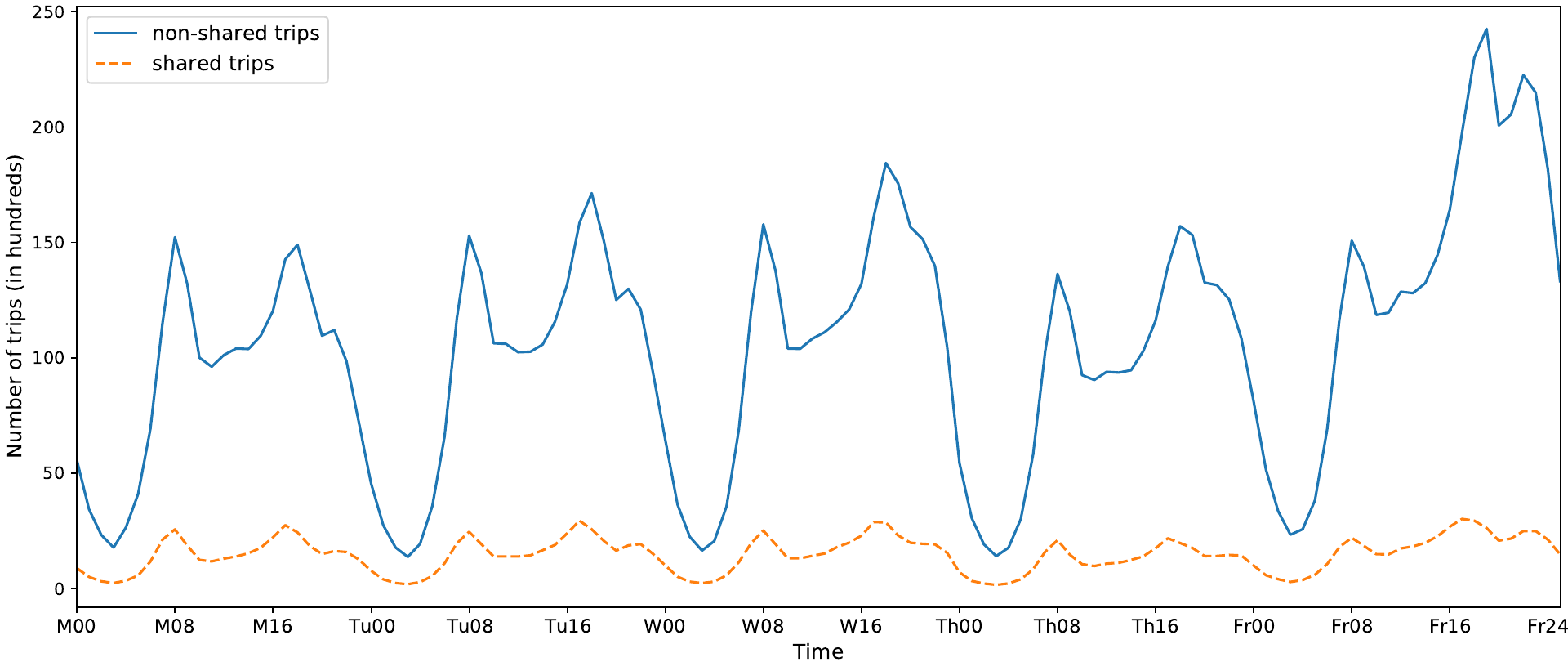}
  \caption{Average numbers of shared and non-shared trips in TNP dataset.}
  \label{fig-sub-originalTrips}
\end{subfigure}%
\hfill
\begin{subfigure}{.37\textwidth}
  \centering
  \includegraphics[width=0.98\linewidth]{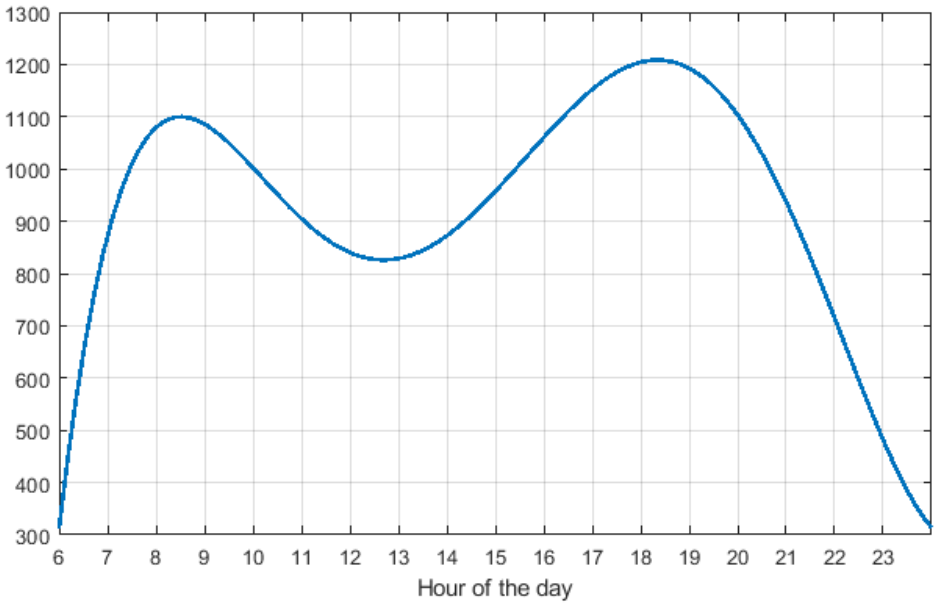}
  \caption{Total number of driver and rider trips generated for each time interval.}
  \label{fig-sub-nTrips}
\end{subfigure}
\caption{Plots for the number of trips for every hour from data and generated.}
\label{fig-nTrips-plot}
\vspace*{-1mm}
\end{figure}
A set of trips are called \emph{shared trips} if this set of trips are matched for the same vehicle consecutively such that their trips may potentially overlap, namely, one or more passengers are in the same vehicle.
The number of shared trips shown in Figure~\ref{fig-sub-originalTrips} suggests that drivers and passengers are willing to share the same vehicle.
For all other trips, we call them \textit{non-shared trips}.
From the plot, the peak hours are between 7:00AM to 10:00AM and 5:00PM to 8:00PM on weekdays for both non-shared and shared trips.
The number of trips generated for each interval roughly follows the function plotted in Figure~(\ref{fig-sub-nTrips}), which is a scaled down and smoothed version of the TNP dataset for weekdays.
For the base instance, the ratio between the number of drivers and riders generated is roughly 1:3 (1 driver and 3 riders) for each interval.
Such a ratio is chosen because it should reflect the system's potential as capacity of 3 is common for most vehicles.
For each time interval, we first generate a set $R$ of riders and then a set $D$ of drivers.
We do not generate a trip where its origin and destination are close. For example, any trip with an origin in Area25 and destination in Area15 is not generated.

\paragraph{Generation of rider trips}
We assume that the numbers of riders entering and exiting a station are roughly the same each day.
Next we assume that the numbers of riders in PTR over the time intervals each day follow a similar distribution of the TNP trips over the time intervals.
Each day is divided into 6 different consecutive time periods (each consists of multiple time intervals):
\textit{morning rush, morning normal, noon, afternoon normal, afternoon rush}, and \textit{evening} time periods.
Each time period determines the probability and distribution of origins and destinations.
Based on the PTR dataset and Rail Capacity Study by CTA~\citep{CTA19}, many riders are going into downtown in the morning and leaving downtown in the afternoon.

For each rider trip $j$ generated, we first randomly decide a pickup area where origin $o_j$ is located within, then decide a dropoff area where destination $d_j$ is located within.
A \emph{pickup area} or \emph{dropoff area} is one of the 22 community areas we selected to build our geological map for the simulation.
For each community area, a set of points spanning the area is defined (each point is represented by a latitude-longitude pair). 
To generate a rider trip $j$ during \textbf{morning rush} time period, the pickup area for $j$ is selected uniformly at random from the list of 22 community areas.
The origin $o_j$ is a point selected uniformly at random from the set of points in the selected pickup area.
Then, we use the standard normal distribution to determine the \emph{dropoff area}, namely, the 22 selected community areas are transformed to follow the standard normal distribution.
Specifically, downtown area is within two SDs (standard deviations), airports are more than two and at most three SDs, and the other urban community areas are more than three SDs away from the mean.
Then, the dropoff area is sampled/selected randomly from this distribution.
The destination $d_j$ is a point selected uniformly at random from the set of points in the selected dropoff area.

The above is repeated until $a_t$ riders are generated, where $a_t + a_t / 3$ (riders + drivers so that it is roughly 1:3 driver-rider ratio) is the total number of trips for time interval $t$ shown in Figure~(\ref{fig-sub-nTrips}).
For any pickup area $c$, let $c_t$ be the number of generated riders originated from $c$ for time interval $t$, that is, $\sum_c c_t = a_t$.
Other time periods follow the same procedure, all urban communities and designated locations can be selected as pickup and dropoff areas.

\begin{enumerate}
\item \textbf{Morning normal} (10:00AM to 12:00AM). For selecting pickup areas, the 22 community areas are transformed to follow the standard normal distribution: urban community areas are within two SDs, downtown is more than two and at most three SDs, and airports are more than three SDs away from the mean; and destination areas are selected using uniform distribution.

\item \textbf{Noon} (12:00PM to 2:00PM). Pickup/dropoff areas are selected uniformly at random from the list of 22 community areas.

\item \textbf{Afternoon normal} (2:00PM to 5:00PM). For selecting pickup areas, downtown and airport are within two SDs and urban community areas are more than two SDs away from the mean.
For selecting dropoff areas, urban community areas are within two SDs, and downtown and airports are more than two SDs away from the mean.

\item \textbf{Afternoon rush} (5:00PM to 8:00PM). For selecting pickup areas, downtown is within two SDs, airports are more than two SDs and at most three SDs, and urban community areas are more than three SDs away from the mean.
For selecting dropoff areas, urban community areas are within two SDs, airports are more than two SDs and at most three SDs, and downtown is more than three SDs away from the mean.

\item \textbf{Evening} (8:00PM to 11:59PM). For both pickup and dropoff areas, urban community areas are within two SDs, downtown is more than two and at most three SDs, and airports are more than three SDs away from the mean.
\end{enumerate}

\paragraph{Generation of driver trips}
We examined the TNP dataset to determine whether, in practice, there are enough drivers who can provide ridesharing service to riders that follow match Types 1 and 2 traffic pattern.
First, we removed any trip from TNP if it is too short (less than 15 minutes or origin and destination are adjacent areas).
We calculated the average number of trips per hour originated from every pre-defined area in the transit network (Figure~\ref{fig-transit-network}), and then plotted the destinations of such trips in a grid heatmap.
In other words, each cell $(c,r)$ in the heatmap represents the the average number of trips per hour originated from area $c$ to destination area $r$ in the transit network (Figure~\ref{fig-transit-network}).
An example is depicted in Figure~\ref{fig-OD-distribution}.
\begin{figure}[!ht]
\centering
\includegraphics[width=\textwidth]{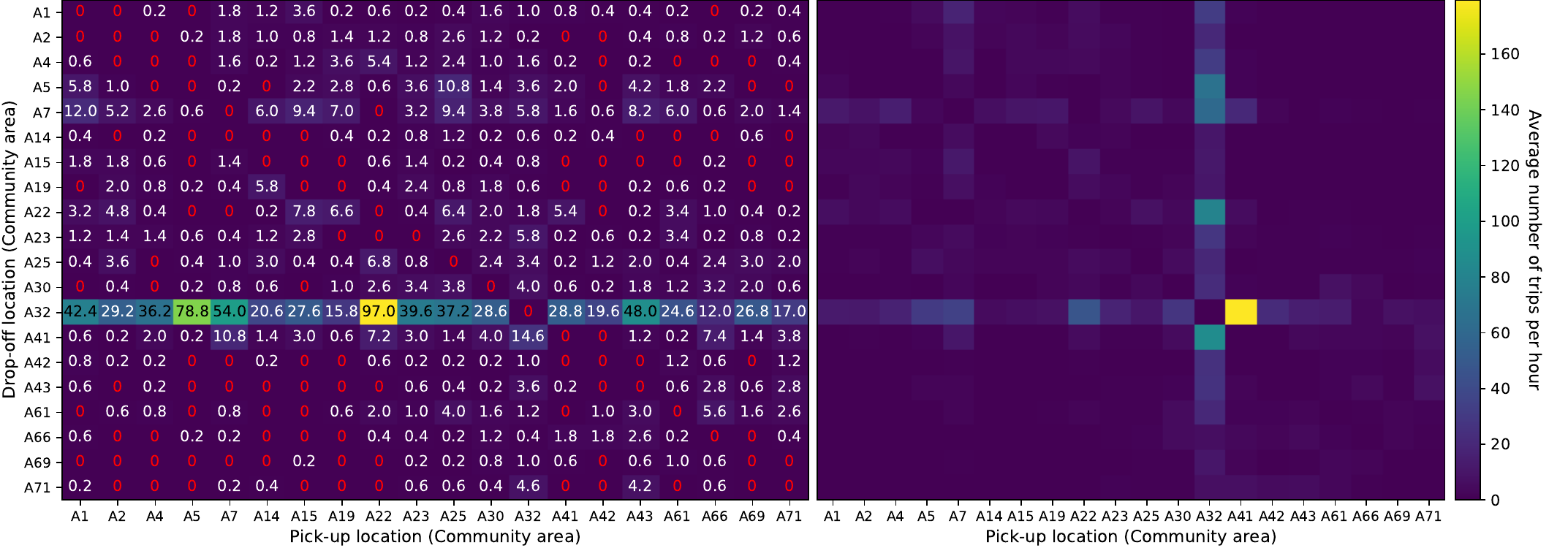}
\caption{Traffic heatmaps for the average number of trips originated from one area (x-axis) during hour 7:00 (left) and hour 17:00 (right) to every other destination area (y-axis).}
\label{fig-OD-distribution}
\end{figure}
From the heatmaps, many trips are going into the downtown area (A32) in the morning; and as time progresses, more and more trips leave downtown. This traffic pattern confirms that there are enough drivers to serve the riders in our simulation.
The number of shared trips shown in Figure~\ref{fig-sub-originalTrips} also suggests that many riders are willing to share a same vehicle.
We slightly reduce the difference between the values of each cell in the heatmaps and use the idea of marginal probability to generate driver trips.
Let $d(c,r,h)$ be the value at the cell $(c,r)$ for origin area $c$, destination $r$ and hour $h$.
Let $P(c,h)$ be sum of the average number of trips originated from area $c$ for hour $h$ (the column for area $c$ in the heatmap corresponds to hour $h$), that is, $P(c, h) = \sum_{r} d(c,r,h)$ is the sum of the values of the whole column $c$ for hour $h$.
Given a time interval $t$, for each area $c$, we generate $c_t/3$ drivers ($c_t$ is defined in Generation of rider trips) such that each driver $i$ has origin $o_i = c$ and destination $d_i = r$ with probability $d(c,r,h)/P(c,h)$, where $t$ is contained in hour $h$.
The probability of selecting an airport as destination is fixed at 5\%.

\paragraph{Deciding other parameters for each trip}
After the origin and destination of a rider or driver trip have been determined, we decide other parameters of the trip.
The capacity $n_i$ of drivers' vehicles is selected from three ranges: the {\em low range} [1,2,3], {\em mid range} [3,4,5], and {\em high range} [4,5,6].
During morning/afternoon peak hours, roughly 95\% and 5\% of vehicles have capacities randomly selected from the low range and mid range, respectively.
It is realistic to assume vehicle capacity is lower for morning and afternoon peak-hour commute.
While during off-peak hours, roughly 80\%, 10\% and 10\% of vehicles have capacities randomly selected from low range, mid range and high range, respectively.
The number $\delta_i$ of stops equals to $n_i$ if $n_i \leq 3$, else it is chosen uniformly at random from $[n_i-2, n_i]$ inclusive.
The detour limit $z_i$ of each driver is within 5 to 20 minutes because traffic is not considered, and transit time and service time are considered in a simplified model.
Earliest departure time $\alpha_i$ of a driver or rider $i$ is from immediate to two time intervals.
Latest arrival time $\beta_i$ of a driver $i$ is at most $1.5 \cdot (t(o_i,d_i) + z_i) + \alpha_i$. 
Latest arrival time $\beta_j$ of a rider $j$ is $\alpha_j + t(\hat{\pi}_j(\alpha_j))$, where $\hat{\pi}_j(\alpha_j)$ is the fastest public transit route for $j$.
The acceptance threshold $\theta_j$ of every rider $j$ is 0.8 for the base instance.
The general information of the base instance is summarized in Table~\ref{table-simulation}.
Note that the earliest departure time all trips generated in the last four time intervals is immediate for computational-result purpose.
\begin{table}[!ht]
\footnotesize
\centering
   \begin{tabular}{ l | p{10.6cm} }
   	\hline
   	Major trip patterns       & from urban communities to downtown and vice versa for peak and off-peak hours, respectively;
trips specify one match type for peak hours and can be in either type for off-peak hours \\
    \# of intervals simulated & Start from 6:00 AM to 11:59 PM; each interval is 15 minutes		\\
   	\# of trips per interval  & varies from [350, 1150] roughly, see Figure~\ref{fig-nTrips-plot}  	\\ 
   	Driver-rider ratio        & 1:3 approximately												\\ 
   	Capacity $n_i$ of vehicles 	  & low: [1,3], mid: [3,5] and high: [4,6] inclusive \\ 
   	Number $\delta_i$ of stops limit 	& $\delta_i=n_i$ if $n_i \leq 3$, or $\delta_i \in [n_i-2,n_i]$ if $n_i \geq 4$		\\
    Earliest departure time $\alpha_i$    & immediate to 2 intervals after a trip $i$ (driver or rider) is generated  \\
    Driver detour limit $z_i$      & 5 minutes to min\{$2 \cdot t(o_i,d_i)$ (driver's fastest route), 20 minutes\}	\\
   	Latest arrival time $\beta_i$ of driver $i$ & $\beta_i \leq 1.5 \cdot (t(o_i,d_i) + z_i) + \alpha_i$                       \\
 	Latest arrival time $\beta_j$ of rider $j$ & $\beta_j = \alpha_j + t(\hat{\pi}_j(\alpha_j))$, where $\hat{\pi}_j(\alpha_j)$ is the fastest public transit route for $j$ with earliest departure time $\alpha_j$ from $o_j$ \\
    Travel duration $\gamma_i$ of driver $i$  & $\gamma_i = t(o_i,d_i) + z_i$   		\\
    Travel duration $\gamma_j$ of rider $j$  & $\gamma_j = t(\hat{\pi}_j(\alpha_j))$, where $\hat{\pi}_j(\alpha_j)$ is the fastest public transit route for $j$   	\\ 
    Acceptance threshold      & 80\% for all riders	(0.8 times the fastest public transit route)           	\\   
    Train and bus travel time   & average at 1.15 and 2 times the fastest route by car, respectively \\ \hline
   \end{tabular}
\captionsetup{font=small}
\caption{General information of the base instance.}
\label{table-simulation}
\end{table}

\paragraph{Reduction configuration procedure.}
When the number of trips increases, the running time for Algorithm~2 and the time needed to construct the $k$-set packing instance (independent set instance) increase significantly. This is due to the increased number of feasible matches for each driver $i \in D$.
In a practical setup, we may restrict the number of feasible matches a driver can have.
Each match produced by Algorithm~1 is called a \emph{base match}, which consists of exactly one driver and one rider.
To make the simulation feasible and practical, we heuristically limit the numbers of base matches for each driver and each rider and the number of total feasible matches for each driver. We use $(x\%, y, z)$, called \emph{reduction configuration} (\emph{Config} for short), to denote that for each driver $i$, the number of base matches of $i$ is reduced to $x$ percentage and at most $y$ total feasible matches are computed for $i$; and for each rider $j$, at most $z$ base matches containing $j$ are used.

After Algorithm~1 is completed. A reduction procedure may be evoked with respect to a Config.
Let $H(V,E)$ be the graph after computing all feasible base matches (instance computed by Algorithm~1 and before Algorithm~2 is executed).
For a trip $i \in \mathcal{A}$, let $E_i$ be the set of base matches of $i$.
The reduction procedure works as follows.
\begin{itemize}
\item First of all, the set of drivers is sorted in descending order of the number of base matches each driver has.

\item Each driver $i$ is then processed one by one.
    \begin{enumerate}
    \item If driver $i$ has at least 10 base matches, then $E_i$ is sorted, based on the number of base matches each rider included in $E_i$ has, in descending order. Otherwise, skip $i$ and process the next driver.
    
    \item For each base match $e=(i,j)$ in the sorted $E_i$, if rider $j$ belongs to $z$ or more other matches, remove $e$ from $E_i$.
    
    \item After above step 2, if $E_i$ has not been reduced to $x\%$, sort the remaining matches in descending order of the travel time from $o_i$ to $o_j$ for remaining matches $e=(\eta_i,r_j)$.
    Remove the first $x'$ matches from $E_i$ until $x\%$ is reached.
    \end{enumerate}
\end{itemize}
The original sorting of the drivers allows us to first remove matches from drivers that have more matches than others.
The sorting of the base matches of driver $i$ in step 1 allows us to first remove matches containing riders that also belong to other matches.
Riders farther away from a driver $i$ may have lower chance to be served together by $i$; this is the reason for the sorting in step 3.

\subsection{Computational results}
We use the same transit network and same set of generated trip data for all algorithms.
All algorithms were implemented in Java, and the experiments were conducted on Intel Core i7-2600 processor with 1333 MHz of 8 GB RAM available to JVM.
To solve the ILP formulation \eqref{obj-1}-\eqref{constraint-2} and the formulation in LPR, we use CPLEX\footnote{IBM ILOG CPLEX v12.10.0}; and we label the algorithm CPLEX uses to solve these ILP formulations by \textbf{Exact}.
Since the optimization goal is to assign acceptable ridesharing routes to as many riders as possible, the performance measure is focused on the number of riders serviceable by acceptable ridesharing routes, followed by the total time saved for the riders as a whole.
We record both of these numbers for each of the algorithms: ImpGreedy, LPR, Exact, Greedy, AnyImp and BestImp.
A rider $j \in R$ is called \emph{served} if $r_j \in \sigma(i)$ for some driver $\eta_i \in D$ such that $\sigma(i)$ belongs to a solution computed by one of the; and we also call such a served rider a \emph{passenger}.

\subsubsection{Results on a base case instance} \label{sec-experiment-base}
The base case instance uses the parameter setting described in Section~\ref{sec-instances} and Config (30\%, 600, 20).
The overall experiment results are shown in Table~\ref{table-base-result}.
\begin{table}[ht]
\scriptsize
\centering
   \begin{tabular}{ l | r | r | r | r | r | r}
   	\hline
& \multicolumn{1}{c|}{ImpGreedy}   & \multicolumn{1}{c|}{LPR}    & \multicolumn{1}{c|}{Exact}	  & \multicolumn{1}{c|}{Greedy}     & \multicolumn{1}{c|}{AnyImp}    & \multicolumn{1}{c}{BestImp}   \\ \hline
  Total number of riders served                & 26597 	  & 22583     & 27940     & 26597      & 27345     & 27360      \\    
  Avg number of riders served per interval     & 369.4	  & 313.7     & 388.1     & 369.4      & 379.8     & 380.0       \\
  Total time saved of all served riders        & 309369.1   & 260427.3  & 324718.4  & 309369.1  & 318729.6  & 318983.9 \\
  Avg time saved of served riders per interval & 4296.8     & 3617.0    & 4510.0    & 4296.8    & 4426.8    & 4430.3     \\
  Avg time saved per served rider & 11.63   & 11.53   & 11.62    & 11.63   & 11.65   & 11.66   \\
  Avg time saved per rider        & 6.68    & 5.75    & 7.17     & 6.68    & 7.03    & 7.04    \\ \hline
	\multicolumn{2}{ l |}{Avg public transit duration per rider} & \multicolumn{5}{l}{30.54 minutes} \\
    \multicolumn{2}{ l |}{Total number of riders and public transit duration} & \multicolumn{5}{l}{45314 and 1384100.97 minutes}\\ \hline
   \end{tabular}
\captionsetup{font=small}
\caption{Base case solution comparison between all algorithms. Every time unit is measured in minute.}
\label{table-base-result}
\end{table}
Although the solutions computed by AnyImp and BestImp are slightly better than that of ImpGreedy, it takes much longer for AnyImp and BestImp to run to completion, as shown in a later experiment (Figure~\ref{fig-configurations-time}).
The average number of riders served per interval is calculated as the total number of riders served divided by 72 (the number of intervals).
The average time saved per served rider is calculated as the total time saved divided by the total number of served riders. 
The results of ImpGreedy and Greedy are aligned since they are essentially the same algorithm: 58.69\% of total number of all riders are assigned acceptable routes and 22.35\% of total time are saved for those riders.
The results of AnyImp and BestImp are similar because of the density of the independent set graph $G(V,E)$ due to Observation~\ref{obs-base-match}.
For AnyImp and BestImp, roughly 60.38\% of total number of all riders are assigned acceptable routes and 23.05\% of total time are saved.
For LPR, 49.8\% of total number of all riders are assigned acceptable routes and 18.82\% of total time are saved. We show that LPR is worse than ImpGreedy in terms of performance and running time in a later experiment.
For Exact, 61.66\% of total number of all riders are assigned acceptable routes and 23.46\% of total time are saved.

The average public transit duration per rider is calculated as the total public transit duration divided by the total number of all riders, which is 30.54 minutes.
The average time saved per rider is calculated as the total time saved divided by the total number of all riders (served and unserved).
From Algorithm Exact, a rider is able to reduce their travel duration from 30.54 minutes to 23.37 minutes on average (save 7.17 minutes) with the integration of public transit and ridesharing.

If we consider only the served riders (26597 for ImpGreedy and 27940 for Exact), the average original public transit duration per served rider is 30.29 (30.30) minutes for ImpGreedy (Exact respectively).
In this case, the average public transit + ridesharing duration per served rider is 18.66 (18.68) minutes, 11.63 (11.62) minutes saved, for ImpGreedy (Exact respectively).
Although this is only a side-effect of the optimization goal of MTR, it reduces riders' travel duration significantly.
The results of these algorithms are not too far apart.
However, it takes too long for AnyImp and BestImp to run to completion.
A 10-second limit is set for both algorithms in each iteration for finding an independent set improvement.
With this time limit, AnyImp and BestImp run to completion within 10 minutes for almost all intervals.
The optimal solution computed by Exact serves only about 5\% more total riders than that of the solution computed by ImpGreedy; and it is most likely constrained by the number of feasible matches each driver has, which is also limited by the base match reduction Config (30\%, 600, 20).
We explore more about this in Section~\ref{sec-experiment-config}.

We also examined the results from the drivers' perspective; we recorded both the mean occupancy rate and vacancy rate of drivers.
The results are depicted in Table~\ref{table-base-OR-VR} and Figure~\ref{fig-OR-VR}.
\begin{table}[bp]
\footnotesize
\centering
   \begin{tabular}{ l | c | c | c | c }
   	\hline
                                            & ImpGreedy 	 &  LPR    & Exact    & BestImp  \\ \hline
    Average occupancy rate per interval     & 2.703       & 2.417    & 2.789    & 2.753    \\
    Average vacancy rate per interval       & 0.0693      & 0.193    & 0.0289   & 0.0436   \\ \hline
   \end{tabular}
\captionsetup{font=small}
\caption{The average occupancy rate and vacancy rate per interval.}
\label{table-base-OR-VR}
\end{table}
\begin{figure}[!bp]
\centering
\includegraphics[width=\textwidth]{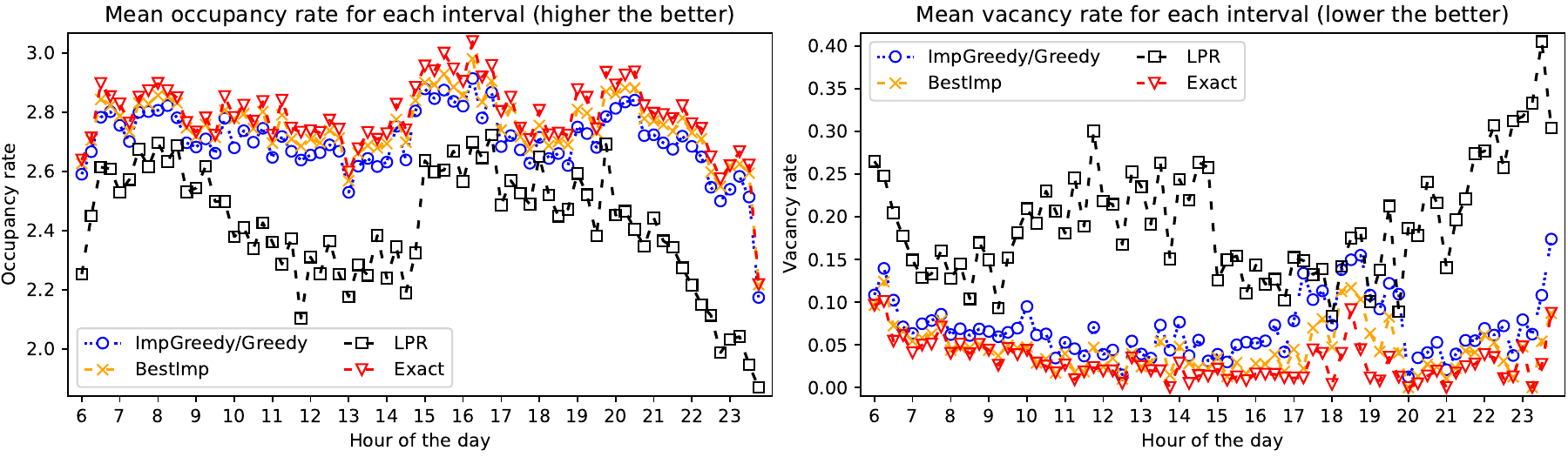}
\captionsetup{font=small}
\caption{The average occupancy rate and vacancy rate of drivers for each interval.}
\label{fig-OR-VR}
\end{figure}%
The mean occupancy rate is calculated as, in each interval, (the number of served passengers + the total number of drivers) divided by the total number of drivers.
The mean vacancy rate describes the number of empty vehicles, so it is calculated as, in each interval, the number of drivers who are not assigned any passenger divided by the total number of drivers.
The average occupancy rate per interval is the sum of mean occupancy rate in each interval divided by the number of intervals (72); and similarly for average vacancy rate per interval.
The occupancy rate results show that in many intervals, 1.7-1.8 passengers are served by each driver on average (except Algorithm LPR).
The vacancy rate results show that in many intervals, only 4-7.5\% and 2-5\% of all the drivers are not assigned any passenger for ImpGreedy and BestImp/Exact, respectively, during all hours except afternoon peak hours.
This is most likely due to the origins of many trips are from the same area (downtown); and if the destinations of drivers and riders do not have the same general direction originated from downtown, the drivers may not be able to serve many riders. On the other hand, when their destinations are aligned, drivers are likely to serve more riders.
The occupancy rate is much lower in the last interval because the number of riders is low, causing the number of served riders low.

\subsubsection{Results on different reduction configurations} \label{sec-experiment-config}
Another major component of the experiment is to measure the performance of the algorithms using different reduction configurations.
We tested 12 different Configs:
\begin{itemize}[leftmargin=*]
\begin{footnotesize}
\item \textit{Small1} (20\%,300,10), \textit{Small2} (20\%,600,10), \textit{Small3} (20\%,300,20), \textit{Small4-10} (20\%,600,20).

\item \textit{Medium1} (30\%,300,10), \textit{Medium2} (30\%,600,10), \textit{Medium3} (30\%,300,20), \textit{Medium4-10} (30\%,600,20).

\item\textit{Large1} (40\%,300,10), \textit{Large2} (40\%,600,10), \textit{Large3-10} (40\%,300,20), and \textit{Large4-10} (40\%,600,20).
\end{footnotesize}
\end{itemize}
Any Config with label ``-10'' at the end means there is a 10-second limit for AnyImp and BestImp to find an independent set improvement (Configs without any label have a 20-second limit).
Note that all 12 Configs have the same sets of driver/rider trips and base matches but have different feasible matches generated at the end (after Algorithm 2).
The performance and running time results of all 12 Configs are depicted in Figures~\ref{fig-configurations-perf}~and~\ref{fig-configurations-time}, respectively.
Since the performance results of ImpGreedy and Greedy are the same, we skip Greedy.
\begin{figure}[b]
\centering
\includegraphics[width=\textwidth]{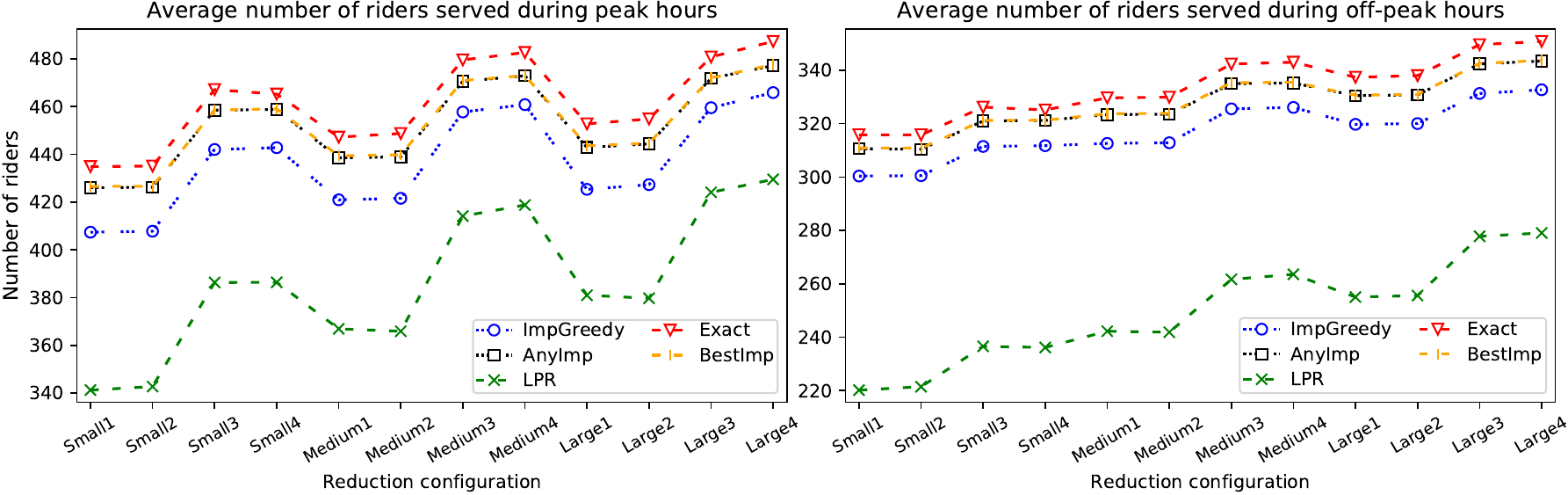}
\captionsetup{font=small}
\caption{Average performance of peak and off-peak hours for different configurations.}
\label{fig-configurations-perf}
\end{figure}
\begin{figure}[t]
\centering
\includegraphics[width=\textwidth]{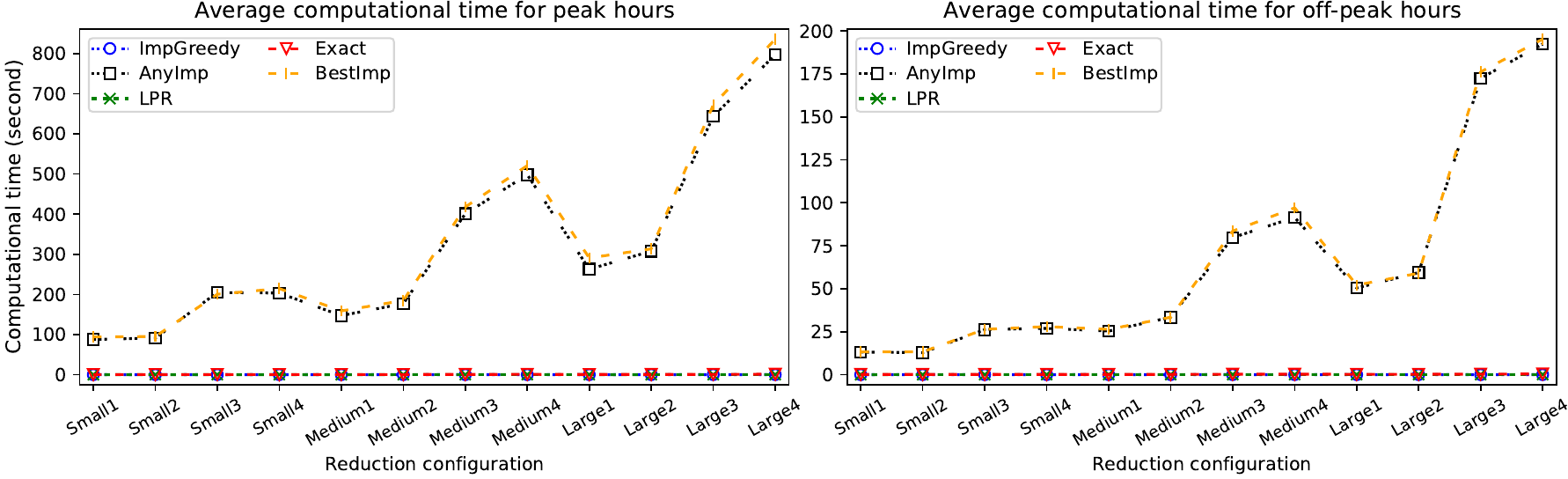}
\captionsetup{font=small}
\caption{Average running time of peak and off-peak hours for different configurations.}
\label{fig-configurations-time}
\end{figure}%

The results are divided into peak and off-peak hours for each Config, averaging all intervals of peak hours and off-peak hours.
As expected, larger Configs give better performance (more passengers are served by drivers).
The increase in performance of Exact, compared to ImpGreedy, remains at about 5\% for each different Config.
This shows that ImpGreedy is practical in terms of performance.
For all algorithms, the increase in performance from Small1 to Small3 is much larger than that from Small1 to Small2 (same for Medium and Large), implying any parameter in a Config should not be too small.
The increase in performance from Large1 to Large4 is higher than that from Medium1 to Medium4 (similarly for Small).
Therefore, a balanced configuration is more important than a configuration emphasizes only one or two parameters.
The average running times of ImpGreedy, LPR and Exact are under a second for all Configs.
On the other hand, for AnyImp and BestImp during peak hours, they require 600-800 seconds and 400-500 seconds for Large3/Large4 and for Medium3/Medium4 Configs, respectively.
By reducing more matches, we are able to improve the running time of AnyImp and BestImp significantly by sacrificing performance slightly.
However, it may still be not practical to use AnyImp and BestImp for peak hours.

We specifically compared the performance of ImpGreedy and LPR since these two are the more practical approximation algorithms.
The performance and running time results of ImpGreedy and LPR using Medium4 and Large4 configs are depicted in Figure~\ref{fig-configurations-perf-time-LPR}.
\begin{figure}[b]
\centering
\includegraphics[width=\textwidth]{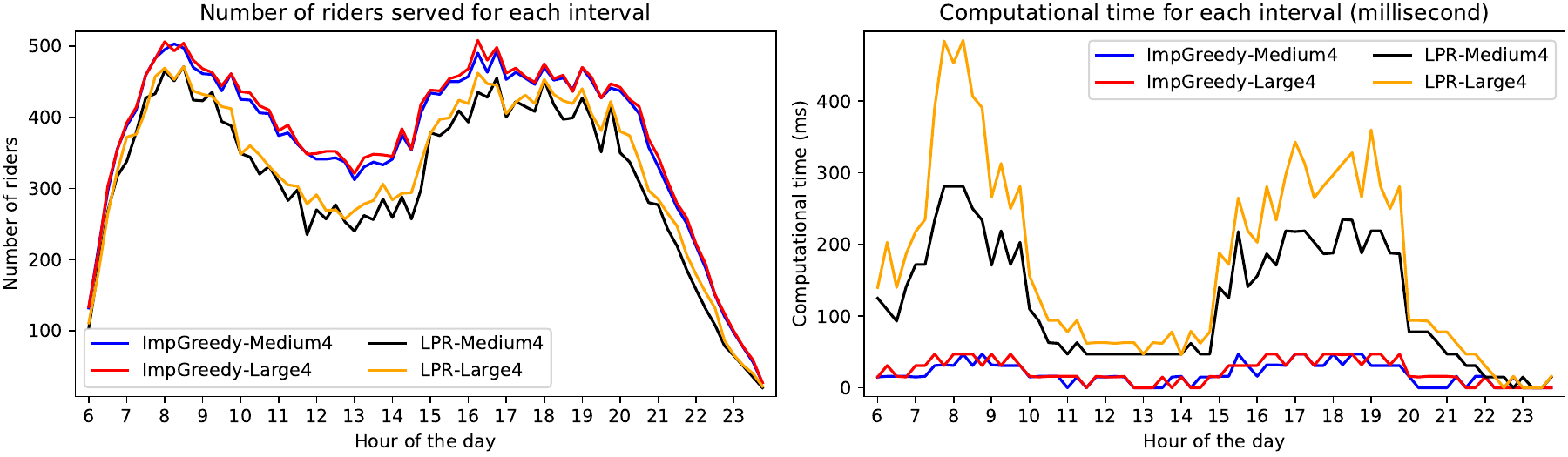}
\captionsetup{font=small}
\caption{Performance of ImpGreedy and LPR for Medium4 and Large4 Configs.}
\label{fig-configurations-perf-time-LPR}
\end{figure}%
For both Configs, ImpGreedy is better than LPR in both performance and running time for each interval.
The difference in performance is most likely due to the removal of riders (step 3 of LPR).
When a rider $j$ is removed from a match $\sigma(i)$, a match $\sigma'(i)$ with $j\notin \sigma'(i)$ and $|\sigma'(i)|> |\sigma(i) \setminus \{j\}|$ for driver $i$ is not searched (even if such a match exists).

We further tested ImpGreedy, Exact and Greedy with the following Configs: \textit{Huge1} (100\%,600,10), \textit{Huge2} (100\%,2500,20) and \textit{Huge3} (100\%,10000,30) (these Configs have the same sets of driver/rider trips and base match sets as those in the previous 12 Configs). The focus of these Configs is to see if these algorithms can handle large number of feasible matches.
The results are shown in Table~\ref{table-hugeConfig}.
\begin{table}[!t]
\scriptsize
\centering
\begin{tabular}{ l | c | c | c }
\hline
	\textbf{ImpGreedy}                                                          & Huge1                  & Huge2                 & Huge3               \\ \hline
	Avg number of riders served for peak/off-peak hours      & 405.8 / 329.2       & 458.7 / 347.5      & 482.5 / 354.2     \\
	Avg time saved of riders for peak/off-peak hours (min)   & 3462.7 / 4500.1   & 3987.6 / 4756.9  & 4237.4 / 4836.9          \\ 
	Avg time saved of riders per interval (min)                     & 4154.3                  & 4500.5                & 4637.0               \\
	Avg running time for peak/off-peak hours (sec)                & 0.0690 / 0.0254       & 0.327 / 0.0824    & 0.806 / 0.170     \\ \hline
	\textbf{Exact}                                                                  & Huge1                   & Huge2                 & Huge3                 \\ \hline
	Avg number of riders served for peak/off-peak hours      & 442.4 / 355.5       & 488.9 / 371.0      & 507.7 / 375.6       \\
	Avg time saved of riders for peak/off-peak hours (min)   & 3701.6 / 4911.7   & 4118.8 / 5128.3  & 4322.9 / 5216.2    \\
	Avg time saved of riders per interval (min)                     & 4508.3                 & 4791.8                 & 4918.4                  \\
	Avg running time for peak/off-peak hours (sec)                & 1.246 / 0.818      & 6.689 / 2.621       & 26.315 / 5.757      \\ \hline
	\textbf{Greedy}                                                                & Huge1                 & Huge2                  & Huge3            \\ \hline
	Avg running time for peak/off-peak hours (sec)                & 10.499 / 2.371    & N/A                     & N/A                \\
	Avg instance size $G(V,E)$ of morning peak ($|E(G)|$)   & 0.014 billion       & 0.25 billion         & 2.4 billion      \\
	Avg time creating $G(V,E)$ of morning peak (sec)           & 10.43                  & 200.41                 & 1391.48         \\ \hline
\end{tabular}
\captionsetup{font=small}
\caption{The results of ImpGreedy and Greedy using Huge Configs.}
\label{table-hugeConfig}
\end{table}

Because ImpGreedy does not create the independent set instance, it runs quicker and uses less memory space than those of Greedy.
Greedy cannot run to completion for Huge2 and Huge3 Configs because in many intervals, the whole graph $G(V,E)$ of the independent set instance is too large to hold in memory (8.00 GB for JVM).
The average numbers of edges in $G(V,E)$ for morning peak hours are 0.014, 0.25 and 2.4 billion for Huge1, Huge2 and Huge3, respectively.
There are techniques in graph processing to solve this problem.
For example, one can use the out-of-core technique, which is to load the needed portion of a graph $G(V,E)$ for processing and unload the processed portion if necessary~\citep[e.g.,][]{Liu-FAST17,Vora-ATC19,Vora-ATC16}.
However, this increases the total running time of the algorithms as the number of I/Os increases.
Another way is to use distributed architecture to process large graphs~\citep{Bouhenni-ACMCS22,Low-VLDBE12}.
This approach may not create a burden in the running time, but it complicates the implementation and maintenance of the system.
More importantly, the time it takes to create $G(V,E)$ can excess practicality in the first place regardless of what technique is used.
The time, displayed in the last row of Table~\ref{table-hugeConfig}, is only the duration for finding all overlapping feasible matches to see if edges of $G(V,E)$ should be created (no actual independent set instance was created for Huge2 and Huge3).

Hence, using Greedy for large instances may not be practical, whereas both ImpGreedy and Exact can handle large instances and can run to completion quickly.
For Huge3, there are 393738 feasible matches on average per interval during peak hours.
ImpGreedy and Exact are able to compute a solution from these many feasible matches in about a second and 26 seconds, respectively.
This shows that both algorithms are scalable when a reasonable Config is used or number of feasible matches is only reasonably large.
Note that the average numbers of served passengers (405.8) and running time (0.06904 seconds) per interval during peak-hours of ImpGreedy with Huge1 is worse than that (452, 0.02475 seconds) produced by ImpGreedy with Medium3.
Similarly, the average numbers of served passengers (442.4) and running time (1.246 seconds) per interval during peak-hours of Exact with Huge1 is worse than that (475.0, 0.6230 seconds) produced by Exact with Medium3.
These support the observation that a balanced configuration is more important than a configuration emphasizes only one or two parameters.

Lastly, we looked at the total (CPU) running times of the algorithms including the time for computing feasible matches (Algorithms 1 and 2).
Table~\ref{table-algorithms-time} shows the average running time of a time interval during peak hours for Algorithm 1 (Alg1), Algorithm 2 (Alg2), and the total time from Algorithm 1 to the finish of each tested algorithm.
\begin{table}[hb]
\footnotesize
\centering
   \begin{tabular}{  l | c || c|c|c|c|c|c }
                   & Alg1  +  Alg2              & \multicolumn{5}{c}{Total time of Alg1 + Alg2 + each algorithm} \\
                    &                                    & ImpGreedy & LPR   & Exact  & Greedy & AnyImp & BestImp   \\ \hline
   {Small3}     & 500.58  + 34.63       & 535.2         & 535.3 & 535.6  & 535.9    & 739.9    & 735.5     \\
   {Small4}      & 500.58  + 35.41      & 536.0         & 536.1 & 536.4  & 536.9    & 739.3    & 750.4     \\
   {Medium3}  & 500.58  + 62.97       & 563.6        & 563.7 & 564.2  & 566.3    & 964.4    & 981.3     \\
   {Medium4}  & 500.58  + 55.32       & 555.9        & 556.1 & 556.7  & 558.8    & 1053.9  & 1076.7   \\
   {Large4}      & 500.58  + 83.26       & 583.9        & 584.1 & 584.9  & 590.3    & 1380.8  & 1419.7   \\
   {Huge2}      & 500.58   + 310.78     & 811.7        & 814.5 & 818.0  & N/A      & N/A       & N/A       \\
   {Huge3}      & 500.58   + 368.93     & 870.3        & 878.1 & 895.8  & N/A      & N/A       & N/A       \\
   \end{tabular}
\captionsetup{font=small}
\caption{Average computational time (in seconds) of an interval during peak hours for all algorithms.}
\label{table-algorithms-time}
\end{table}
The running time of Alg1 solely depends on computing the shortest paths between the trips and stations.
Alg1 runs to completion in about 500 seconds on average per interval during peak hours (7AM-10AM and 5PM-8PM).
As for Algorithm 2, when many trips' origins/destinations are concentrated in one area, the running time increases significantly, especially for drivers with high capacity.
Running time of Alg2 can be reduced significantly by Configs with aggressive reductions.
ImpGreedy and Exact are capable of handling large instances tested. Exact provides better solutions than any of the approximation algorithms. ImpGreedy gives solutions with quality close to other algorithms with running time less than a second for instances tested. ImpGreedy may be more practical for instances larger than those tested.

In conclusion, both ImpGreedy and Exact are much faster and uses less memory space, thus can handle large instances, compared to the other approximation algorithms.
From the experiment results in Figure~\ref{fig-configurations-perf}~and Table~\ref{table-algorithms-time}, it is beneficial to dynamically select different reduction configurations for each interval depending on the number of trips and the number of feasible matches.
When the size of an instance is large and a solution must be computed within some time-limit, ImpGreedy may have a slight advantage over the Exact algorithm.
Recall that the MTR problem (the ILP formulation \eqref{obj-1}-\eqref{constraint-2}) is NP-hard by Theorem~\ref{theorem-nphard} (Theorem~\ref{theorem-ILP}).
From this, if the size of an instance or the number of feasible matches is larger, the running time of Exact for computing an optimal solution is not known and can be time consuming.
As indicated by the results of Huge3, the running time of Exact is 26 times higher than that of ImpGreedy.
A fallback plan would be to run ImpGreedy after Exact.
If after a pre-defined time limit is reached and Exact still cannot compute an optimal solution, the solution computed by ImpGreedy can be used.
As shown by the experiments, the performance of ImpGreedy is still competitive.

\subsubsection{Effects from different acceptance thresholds}
We consider three different acceptance thresholds for riders: 0.9, 0.7 and 0.6 (in addition to 0.8, specified in Table~\ref{table-simulation}, that is already tested in the base instance).
As a reminder, an acceptance threshold (\emph{AT} for short) 0.9 means that the acceptable ridesharing route given to every rider $j$ has travel time (duration) at most 0.9 times $j$'s public transit duration $t(\hat{\pi}_j(\alpha_j))$.
All other parameters in the base instance remain the same.
To see the effect of different acceptance thresholds, two Configs are used: Large4-(40\%, 600, 20) and Huge3-(100\%, 10000, 30).
Only ImpGreedy and Exact were tested.

The overall results are shown in Table~\ref{table-result-ar-large} for Config Large4 and Table~\ref{table-result-ar-huge} for Config Huge3.
\begin{table}[!bp]
\footnotesize
\centering
   \begin{tabular}{ l | r | r | r | r }
   \hline
   \textbf{ImpGreedy}   & \multicolumn{1}{c|}{AT:0.9}   & \multicolumn{1}{c|}{AT:0.8}    & \multicolumn{1}{c|}{AT:0.7}   & \multicolumn{1}{c}{AT:0.6}         \\ \hline
   Total number of riders served                   & 27712     & 27008      & 26099     & 23456       \\    
   Avg number of riders served per interval        & 384.9     & 375.1      & 362.5     & 325.8       \\
   Total time saved of all served riders (minute)	 & 275329.5  & 315851.8   & 344727.3  & 354368.0    \\
   Avg time saved of served riders per interval (minute)	  & 3824.0     & 4386.8   & 4787.9    & 4921.8  \\
   Avg time saved per served rider (minute)		 & 9.94      & 11.69      & 13.21     & 15.11   \\
   Avg time saved per rider (minute)               & 6.07      & 6.97       & 7.61      & 7.82    \\ \hline
   \textbf{Exact}     & \multicolumn{1}{c|}{AT:0.9}     & \multicolumn{1}{c|}{AT:0.8}    & \multicolumn{1}{c|}{AT:0.7}   & \multicolumn{1}{c}{AT:0.6}  \\ \hline
   Total number of riders served                   & 29176     & 28430      & 27561      & 25184          \\    
   Avg number of riders served per interval        & 405.2     & 394.9      & 382.8      & 349.8          \\
   Total time saved of all served riders (minute)	 & 287628.5  & 331979.0   & 364905.9   & 379767.2       \\
   Avg time saved of served riders per interval (minute)	  & 3994.8   & 4610.8   & 5068.1    & 5274.5     \\
   Avg time saved per served rider (minute)		 & 9.86      & 11.68      & 13.24     & 15.08   \\
   Avg time saved per rider (minute)               & 6.35      & 7.33       & 8.05       & 8.38           \\ \hline
   
   \multicolumn{2}{| l |}{Total number of riders and public transit duration} & \multicolumn{3}{l |}{45314 and 1384100.97 minutes}\\ \hline
   \end{tabular}
\captionsetup{font=small}
\caption{Overall solution comparison between different acceptance thresholds for Large3 Config.}
\label{table-result-ar-large}
\end{table}
\begin{table}[!ht]
\footnotesize
\centering
   \begin{tabular}{ l | r | r | r | r}
   \hline
   \textbf{ImpGreedy}    & \multicolumn{1}{c|}{AT:0.9}   & \multicolumn{1}{c|}{AT:0.8}   & \multicolumn{1}{c|}{AT:0.7}   & \multicolumn{1}{c}{AT:0.6}        \\ \hline
   Total number of riders served                  & 29657     & 28579     & 27218      & 24655         \\    
   Avg number of riders served per interval       & 411.9     & 396.9     & 378.0      & 342.4         \\
   Total time saved of all served riders (minute)	& 295310.4  & 333864.3  & 361674.1   & 370964.7      \\
   Avg time saved of served riders per interval (minute)	 & 4101.5   & 4637.0    & 5023.3   & 5152.3   \\
   Avg time saved per served rider (minute)	    & 9.96      & 11.68      & 13.29     & 15.05   \\
   Avg time saved per rider (minute)              & 6.52      & 7.37      & 7.98       & 8.19          \\ \hline
   \textbf{Exact}       & \multicolumn{1}{c|}{AT:0.9}    & \multicolumn{1}{c|}{AT:0.8}   & \multicolumn{1}{c|}{AT:0.7}   & \multicolumn{1}{c}{AT:0.6}  \\ \hline
   Total number of riders served                  & 31168     & 30214     & 29122     & 26973          \\    
   Avg number of riders served per interval       & 432.9     & 419.6     & 404.5     & 374.6          \\
   Total time saved of all served riders (minute)	& 305837.6  & 354127.2  & 386819.1  & 405339.6       \\
   Avg time saved of served riders per interval (minute)	 & 4247.7   & 4918.4    & 5372.5   & 5629.7    \\ 
   Avg time saved per served rider (minute)	    & 9.81      & 11.72      & 13.28     & 15.03   \\
   Avg time saved per rider (minute)              & 6.75      & 7.81      & 8.54      & 8.95           \\ \hline
   \multicolumn{2}{| l |}{Total number of riders and public transit duration} & \multicolumn{3}{l |}{45314 and 1384100.97 minutes}\\ \hline
   \end{tabular}
\captionsetup{font=small}
\caption{Overall solution comparison between different acceptance thresholds for Huge3 Config.}
\label{table-result-ar-huge}
\end{table}
The results for Large4 and Huge3 are consistent for both ImpGreedy and Exact.
As somewhat expected, the total number of riders served decreases for both Configs as the acceptance threshold decreases, since shorter travel duration is required according to riders' requests.
The data in Table~\ref{table-result-ar-large} and Table~\ref{table-result-ar-huge} show that the total time saved of all served riders increases when the acceptance threshold decreases, which implies that the number of riders served is inversely related to the total time saved when the acceptance threshold changes.

Let us focus on the results using Huge3.
From AT 0.9 to 0.8, total number of riders served decreases by 3.635\% (3.061\%), whereas total time saved of all served riders increases by 13.055\% (15.789\%) for ImpGreedy (Exact respectively).
From this, decreasing the acceptance threshold from 0.9 to 0.8 only reduces the number of served riders slightly but significantly reduces the travel time for each served rider.
For a smaller AT, a quicker acceptable route is required for a rider. This may reduces the total number of riders served but each served rider saves more time.
As a result, the average time saved per (served) rider increases as AT decreases. 
Because the optimization goal of the MTR problem is to maximize the number of riders served, the algorithms for MTR find solutions with more riders served instead of focusing on more time saved even if the solutions for a smaller AT are solutions for a greater AT.
From AT 0.8 to 0.7, the gap of inverse relation reduces to: 4.762\% (3.164\%) decreased in total number of riders served and 8.330\% (9.232\%) increased in total time saved of all served riders for ImpGreedy (Exact respectively); and the gap reduces further from AT 0.7 to 0.6.
From the results, it seems that AT between 0.7-0.8 has a nice balance between the number of riders served and time saved of served riders.
Although riders can choose their own acceptance thresholds in practice, the system can suggest a default AT for all riders which would balance between the chance of being served and the amount of time saved.

\begin{figure}[!ht]
\centering
\includegraphics[width=0.88\textwidth]{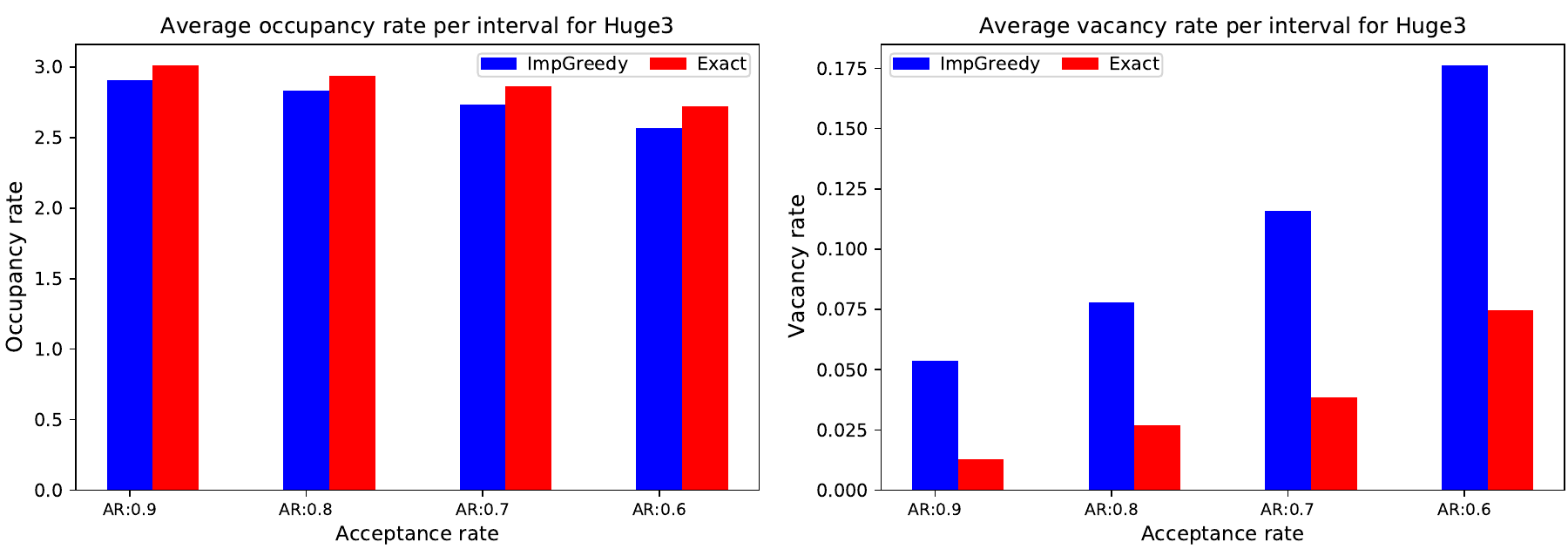}
\captionsetup{font=small}
\caption{The average occupancy rate and vacancy rate per interval for Huge3 Config.}
\label{figure-ar-OR-VR}
\end{figure}
The occupancy rate and vacancy rate of drivers for Huge3 are depicted in Figure~\ref{figure-ar-OR-VR}.
The average occupancy and vacancy rates for both algorithms align with the result shown in Table~\ref{table-result-ar-huge}.
The average occupancy rate for Exact with AT 0.9 actually just exceeds 3 people per vehicle.
As shown, the vacancy rate increases more as the acceptance threshold decreases.
The average vacancy rate for Exact with AT 0.9 is 1.3\%.
If the main goal is to increase occupancy rate and decrease vacancy rate, using a centralized AT of 0.9 is effective.
Bases on this and previous result (Table~\ref{table-result-ar-huge}), an AT of 0.8 seems to be the most balanced.
As a summary of the performance of Exact with AT 0.8 and Config Huge3, 
66.68\% of total riders are assigned ridesharing routes and 25.59\% of total time are saved; and riders are able to reduce their average travel duration from 30.54 minutes to 22.73 minutes.
If we consider only the served riders (30214), the average origin public transit duration per served rider is 30.29 minutes, and the average public transit + ridesharing duration per served rider is 18.57 minutes.

\section{Conclusion and future work} \label{sec-conclusion}
We propose an ILP formulation for the MTR problem, which maximizes the number of public transit riders assigned to drivers subject to shorter commuting time in a transportation system that integrates public transit and ridesharing. We prove that the MTR problem is NP-hard; and we present an exact algorithm (called Exact) and a number of approximation algorithms for the problem. Two of which are very practical: one (called LPR) with $(1-\frac{1}{e})$-approximation ratio using LP relaxation and rounding technique and the other (called ImpGreedy) with $\frac{1}{2}$-approximation ratio  using a greedy approach.
Although Algorithm Exact may run in exponential time in a worst case, experiments show that it is efficient on practical data if the instance is not substantially large.
Algorithm ImpGreedy runs much faster than Exact and has a performance close to Exact.
Despite the theoretical approximation ratio, ImpGreedy outperforms LPR in all instances tested for both performance and running time.
Based on real-world transit datasets in Chicago, our study has shown that integrating public transit and ridesharing can benefit the transportation system as a whole. 
Our base case experiments show that, on average, 61.7\% and 58.7\% of the passengers are assigned ridesharing routes and able to save 23.5\% and 22.4\% of travel time by Exact and ImpGreedy, respectively. 
Majority of the drivers are assigned at least one passenger, and vehicle occupancy rate has improved close to 3 (including the driver) on average.
These results suggest that ridesharing can be a complement to public transit.

The number of passengers assigned to drivers and the time saved by ImpGreedy is about 95\% of those by Exact.
Algorithm ImpGreedy has a polynomial running time in the worst case and runs much faster than Exact for every instance tested.
ImpGreedy can be a fallback plan for Exact when the latter can not give a solution within a time limit in practice.
It is worth improving the performance of ImpGreedy while keeping its time efficiency.
For the instances tested, the time to find feasible matches (Algorithm 1 + Algorithm 2) is much longer than the running time of Exact and ImpGreedy to assign passengers to drivers.
Another future work is to develop faster algorithms for computing feasible matches.

\subsection*{Acknowledgments}
We thank anonymous reviewers for their constructive comments, and we also thank an anonymous reader for suggesting the Generalized Assignment Problem.
This work was partially supported by Canada NSERC Discovery Grants 253500.

\bibliographystyle{elsarticle-harv}
\bibliography{mtr}

\begin{thebibliography}{52}
\expandafter\ifx\csname natexlab\endcsname\relax\def\natexlab#1{#1}\fi
\providecommand{\url}[1]{\texttt{#1}}
\providecommand{\href}[2]{#2}
\providecommand{\path}[1]{#1}
\providecommand{\DOIprefix}{doi:}
\providecommand{\ArXivprefix}{arXiv:}
\providecommand{\URLprefix}{URL: }
\providecommand{\Pubmedprefix}{pmid:}
\providecommand{\doi}[1]{\href{http://dx.doi.org/#1}{\path{#1}}}
\providecommand{\Pubmed}[1]{\href{pmid:#1}{\path{#1}}}
\providecommand{\bibinfo}[2]{#2}
\ifx\xfnm\relax \def\xfnm[#1]{\unskip,\space#1}\fi
\bibitem[{Agatz et~al.(2011)Agatz, Erera, Savelsbergh and Wang}]{Agatz-TRBM11}
\bibinfo{author}{Agatz, N.}, \bibinfo{author}{Erera, A.},
  \bibinfo{author}{Savelsbergh, M.}, \bibinfo{author}{Wang, X.},
  \bibinfo{year}{2011}.
\newblock \bibinfo{title}{Dynamic ride-sharing: A simulation study in metro
  atlanta}.
\newblock \bibinfo{journal}{Transportation Research Part B: Methodological}
  \bibinfo{volume}{45}, \bibinfo{pages}{1450--1464}.
\newblock \DOIprefix\doi{10.1016/j.trb.2011.05.017}.
\bibitem[{Aissat and Varone(2015)}]{Aissat-ICEIS15}
\bibinfo{author}{Aissat, K.}, \bibinfo{author}{Varone, S.},
  \bibinfo{year}{2015}.
\newblock \bibinfo{title}{Carpooling as complement to multi-modal
  transportation}, in: \bibinfo{booktitle}{Enterprise Information Systems
  (ICEIS 2015)}, pp. \bibinfo{pages}{236--255}.
\newblock \DOIprefix\doi{10.1007/978-3-319-29133-8_12}.
\bibitem[{Alonso-Gonz\'{a}lez et~al.(2018)Alonso-Gonz\'{a}lez, Liu, Cats, Oort
  and Hoogendoorn}]{Alonso-Gonzalez-TRR18}
\bibinfo{author}{Alonso-Gonz\'{a}lez, M.}, \bibinfo{author}{Liu, T.},
  \bibinfo{author}{Cats, O.}, \bibinfo{author}{Oort, N.V.},
  \bibinfo{author}{Hoogendoorn, S.}, \bibinfo{year}{2018}.
\newblock \bibinfo{title}{The potential of demand-responsive transport as a
  complement to public transport: an assessment framework and an empirical
  evaluation}.
\newblock \bibinfo{journal}{Transportation Research Record}
  \bibinfo{volume}{2672}, \bibinfo{pages}{879--889}.
\newblock \DOIprefix\doi{10.1177/0361198118790842}.
\bibitem[{Alonso-Mora et~al.(2017)Alonso-Mora, Samaranayake, Wallar, Frazzoli
  and Rus}]{Alonso-Mora-PNAS17}
\bibinfo{author}{Alonso-Mora, J.}, \bibinfo{author}{Samaranayake, S.},
  \bibinfo{author}{Wallar, A.}, \bibinfo{author}{Frazzoli, E.},
  \bibinfo{author}{Rus, D.}, \bibinfo{year}{2017}.
\newblock \bibinfo{title}{On-demand high-capacity ride-sharing via dynamic
  trip-vehicle assignment}.
\newblock \bibinfo{journal}{Proceedings of the National Academy of Sciences}
  \bibinfo{volume}{114}, \bibinfo{pages}{462--467}.
\newblock \DOIprefix\doi{10.1073/pnas.1611675114}.
\bibitem[{{American Public Transportation Association}(2023)}]{APTA22}
\bibinfo{author}{{American Public Transportation Association}},
  \bibinfo{year}{2023}.
\newblock \bibinfo{title}{{2022 Public Transportation Fact Book}}.
\newblock \bibinfo{note}{Technical report}.
\bibitem[{Berman(2000)}]{Berman-SWAT00}
\bibinfo{author}{Berman, P.}, \bibinfo{year}{2000}.
\newblock \bibinfo{title}{A d/2 approximation for maximum weight independent
  set in d-claw free graphs}, in: \bibinfo{booktitle}{Algorithm Theory - SWAT
  2000}, \bibinfo{publisher}{Springer Berlin Heidelberg}. pp.
  \bibinfo{pages}{214--219}.
\newblock \DOIprefix\doi{10.1007/3-540-44985-X_19}.
\bibitem[{Bouhenni et~al.(2021)Bouhenni, Yahiaoui, Nouali-Taboudjemat and
  Kheddouci}]{Bouhenni-ACMCS22}
\bibinfo{author}{Bouhenni, S.}, \bibinfo{author}{Yahiaoui, S.},
  \bibinfo{author}{Nouali-Taboudjemat, N.}, \bibinfo{author}{Kheddouci, H.},
  \bibinfo{year}{2021}.
\newblock \bibinfo{title}{A survey on distributed graph pattern matching in
  massive graphs}.
\newblock \bibinfo{journal}{ACM Comput. Surv.} \bibinfo{volume}{54}.
\newblock \DOIprefix\doi{10.1145/3439724}.
\bibitem[{B\"{u}rstlein et~al.(2021)B\"{u}rstlein, L\'{o}pez and
  Farooq}]{Burstlein-TRAPP21}
\bibinfo{author}{B\"{u}rstlein, J.}, \bibinfo{author}{L\'{o}pez, D.},
  \bibinfo{author}{Farooq, B.}, \bibinfo{year}{2021}.
\newblock \bibinfo{title}{Exploring first-mile on-demand transit solutions for
  north american suburbia: A case study of markham, canada}.
\newblock \bibinfo{journal}{Transportation Research Part A: Policy and
  Practice} \bibinfo{volume}{153}, \bibinfo{pages}{261--283}.
\newblock \DOIprefix\doi{10.1016/j.tra.2021.08.018}.
\bibitem[{{Center for Sustainable Systems, University of
  Michigan}(2020)}]{CSS20}
\bibinfo{author}{{Center for Sustainable Systems, University of Michigan}},
  \bibinfo{year}{2020}.
\newblock \bibinfo{title}{Personal transportation factsheet}.
\newblock \URLprefix
  \url{http://css.umich.edu/factsheets/personal-transportation-factsheet}.
  \bibinfo{note}{pub. No. CSS01-07}.
\bibitem[{Chandra and Halld{\'{o}}rsson(2001)}]{Chandra-JoA01}
\bibinfo{author}{Chandra, B.}, \bibinfo{author}{Halld{\'{o}}rsson, M.},
  \bibinfo{year}{2001}.
\newblock \bibinfo{title}{Greedy local improvement and weighted set packing
  approximation}.
\newblock \bibinfo{journal}{Journal of Algorithms} \bibinfo{volume}{39},
  \bibinfo{pages}{223--240}.
\newblock \DOIprefix\doi{10.1006/jagm.2000.1155}.
\bibitem[{Chen and Wang(2018)}]{Chen-TRBM18}
\bibinfo{author}{Chen, Y.}, \bibinfo{author}{Wang, H.}, \bibinfo{year}{2018}.
\newblock \bibinfo{title}{Pricing for a last-mile transportation system}.
\newblock \bibinfo{journal}{Transportation Research Part B: Methodological}
  \bibinfo{volume}{107}, \bibinfo{pages}{57--69}.
\newblock \DOIprefix\doi{10.1016/j.trb.2017.11.008}.
\bibitem[{{Chicago Transit Authority}(2017)}]{CTA19}
\bibinfo{author}{{Chicago Transit Authority}}, \bibinfo{year}{2017}.
\newblock \bibinfo{title}{System-wide rail capacity study}.
\newblock \URLprefix
  \url{https://www.transitchicago.com/assets/1/6/RP_CDMSMITH_RCM_Task2AExecutiveSummary_20170628_FINAL.pdf}.
  \bibinfo{note}{revised February 2019}.
\bibitem[{Cordeau(2003)}]{Cordeau-TRBM03}
\bibinfo{author}{Cordeau, J.F.}, \bibinfo{year}{2003}.
\newblock \bibinfo{title}{A tabu search heuristic for the static multi-vehicle
  dial-a-ride problem}.
\newblock \bibinfo{journal}{Transportation Research Part B: Methodological}
  \bibinfo{volume}{37}, \bibinfo{pages}{579--594}.
\newblock \DOIprefix\doi{https://doi.org/10.1016/S0191-2615(02)00045-0}.
\bibitem[{Dawande et~al.(2000)Dawande, Kalagnanam, Keskinocak, Salman and
  Ravi}]{Dawande-JCO00}
\bibinfo{author}{Dawande, M.}, \bibinfo{author}{Kalagnanam, J.},
  \bibinfo{author}{Keskinocak, P.}, \bibinfo{author}{Salman, S.},
  \bibinfo{author}{Ravi, R.}, \bibinfo{year}{2000}.
\newblock \bibinfo{title}{Approximation algorithms for the multiple knapsack
  problem with assignment restrictions}.
\newblock \bibinfo{journal}{Journal of Combinatorial Optimization}
  \bibinfo{volume}{4}, \bibinfo{pages}{171--186}.
\newblock \DOIprefix\doi{10.1023/A:1009894503716}.
\bibitem[{Diao et~al.(2021)Diao, Kong and Zhao}]{Diao-NS21}
\bibinfo{author}{Diao, M.}, \bibinfo{author}{Kong, H.}, \bibinfo{author}{Zhao,
  J.}, \bibinfo{year}{2021}.
\newblock \bibinfo{title}{Impacts of transportation network companies on urban
  mobility}.
\newblock \bibinfo{journal}{Nature Sustainability} \bibinfo{volume}{4},
  \bibinfo{pages}{494--500}.
\newblock \DOIprefix\doi{10.1038/s41893-020-00678-z}.
\bibitem[{Feigon and Murphy(2016)}]{Feigon-TRB16}
\bibinfo{author}{Feigon, S.}, \bibinfo{author}{Murphy, C.},
  \bibinfo{year}{2016}.
\newblock \bibinfo{title}{Shared mobility and the transformation of public
  transit}.
\newblock \bibinfo{type}{{TCRP Research Report}}. Transportation Research
  Board.
\newblock \DOIprefix\doi{10.17226/23578}.
\bibitem[{Fleischer et~al.(2006)Fleischer, Goemans, Mirrokni and
  Sviridenko}]{Fleischer-SODA06}
\bibinfo{author}{Fleischer, L.}, \bibinfo{author}{Goemans, M.X.},
  \bibinfo{author}{Mirrokni, V.S.}, \bibinfo{author}{Sviridenko, M.},
  \bibinfo{year}{2006}.
\newblock \bibinfo{title}{Tight approximation algorithms for maximum general
  assignment problems}, in: \bibinfo{booktitle}{Proceedings of the Seventeenth
  Annual ACM-SIAM Symposium on Discrete Algorithm}, \bibinfo{publisher}{Society
  for Industrial and Applied Mathematics}. p. \bibinfo{pages}{611–620}.
\bibitem[{Garey and Johnson(1979)}]{Garey79}
\bibinfo{author}{Garey, M.}, \bibinfo{author}{Johnson, D.},
  \bibinfo{year}{1979}.
\newblock \bibinfo{title}{Computers and Intractability: A Guide to the Theory
  of NP-Completeness}.
\newblock \bibinfo{publisher}{W.H. Freeman and Company}.
\bibitem[{Ghoseiri et~al.(2011)Ghoseiri, Haghani and Hamedi}]{Ghoseiri-USDT11}
\bibinfo{author}{Ghoseiri, K.}, \bibinfo{author}{Haghani, A.},
  \bibinfo{author}{Hamedi, M.}, \bibinfo{year}{2011}.
\newblock \bibinfo{title}{Real-time rideshare matching problem}.
\newblock \bibinfo{type}{Final {Report} of {UMD-2009-04}}. U.S. Department of
  Transportation.
\newblock \URLprefix \url{https://rosap.ntl.bts.gov/view/dot/25988}.
\bibitem[{Gu and Liang(2021)}]{Gu-ISAAC21}
\bibinfo{author}{Gu, Q.}, \bibinfo{author}{Liang, J.}, \bibinfo{year}{2021}.
\newblock \bibinfo{title}{Multimodal transportation with ridesharing of
  personal vehicles}, in: \bibinfo{booktitle}{32nd International Symposium on
  Algorithms and Computation (ISAAC 2021)}, \bibinfo{publisher}{Schloss
  Dagstuhl -- Leibniz-Zentrum f{\"u}r Informatik}. pp.
  \bibinfo{pages}{39:1--39:16}.
\newblock \DOIprefix\doi{10.4230/LIPIcs.ISAAC.2021.39}.
\bibitem[{Gu et~al.(2021)Gu, Liang and Zhang}]{Gu-TCS21}
\bibinfo{author}{Gu, Q.}, \bibinfo{author}{Liang, J.L.},
  \bibinfo{author}{Zhang, G.}, \bibinfo{year}{2021}.
\newblock \bibinfo{title}{Approximate ridesharing of personal vehicles
  problem}.
\newblock \bibinfo{journal}{Theoretical Computer Science}
  \bibinfo{volume}{871}, \bibinfo{pages}{30--50}.
\newblock \DOIprefix\doi{10.1016/j.tcs.2021.04.009}.
\bibitem[{Hazan et~al.(2006)Hazan, Safra and Schwartz}]{Hazan-CC06}
\bibinfo{author}{Hazan, E.}, \bibinfo{author}{Safra, S.},
  \bibinfo{author}{Schwartz, O.}, \bibinfo{year}{2006}.
\newblock \bibinfo{title}{On the complexity of approximating k-set packing}.
\newblock \bibinfo{journal}{Computational Complexity} \bibinfo{volume}{15},
  \bibinfo{pages}{20--39}.
\newblock \DOIprefix\doi{10.1007/s00037-006-0205-6}.
\bibitem[{Henao and Marshall(2021)}]{Henao-Trans19}
\bibinfo{author}{Henao, A.}, \bibinfo{author}{Marshall, W.},
  \bibinfo{year}{2021}.
\newblock \bibinfo{title}{The impact of ride-hailing on vehicle miles
  traveled}.
\newblock \bibinfo{journal}{Transportation} \bibinfo{volume}{49},
  \bibinfo{pages}{2173--2194}.
\newblock \DOIprefix\doi{10.1007/s11116-018-9923-2}.
\bibitem[{Huang et~al.(2019)Huang, Bucher, Kissling, Weibel and
  Raubal}]{Huang-TITS19}
\bibinfo{author}{Huang, H.}, \bibinfo{author}{Bucher, D.},
  \bibinfo{author}{Kissling, J.}, \bibinfo{author}{Weibel, R.},
  \bibinfo{author}{Raubal, M.}, \bibinfo{year}{2019}.
\newblock \bibinfo{title}{Multimodal route planning with public transport and
  carpooling}.
\newblock \bibinfo{journal}{IEEE Transactions on Intelligent Transportation
  Systems} \bibinfo{volume}{20}, \bibinfo{pages}{3513--3525}.
\newblock \DOIprefix\doi{10.1109/TITS.2018.2876570}.
\bibitem[{Karp(1972)}]{Karp72}
\bibinfo{author}{Karp, R.M.}, \bibinfo{year}{1972}.
\newblock \bibinfo{title}{Reducibility among Combinatorial Problems}.
  \bibinfo{publisher}{Springer US}, \bibinfo{address}{Boston, MA}.
\newblock pp. \bibinfo{pages}{85--103}.
\newblock \DOIprefix\doi{10.1007/978-1-4684-2001-2_9}.
\bibitem[{Kumar and Khani(2021)}]{Kumar-TRCET21}
\bibinfo{author}{Kumar, P.}, \bibinfo{author}{Khani, A.}, \bibinfo{year}{2021}.
\newblock \bibinfo{title}{An algorithm for integrating peer-to-peer ridesharing
  and schedule-based transit system for first mile/last mile access}.
\newblock \bibinfo{journal}{Transportation Research Part C: Emerging
  Technologies} ,
  \bibinfo{pages}{122}\DOIprefix\doi{10.1016/j.trc.2020.102891}.
\bibitem[{Liu and Huang(2017)}]{Liu-FAST17}
\bibinfo{author}{Liu, H.}, \bibinfo{author}{Huang, H.H.}, \bibinfo{year}{2017}.
\newblock \bibinfo{title}{Graphene: {Fine-Grained} {IO} management for graph
  computing}, in: \bibinfo{booktitle}{15th USENIX Conference on File and
  Storage Technologies (FAST 17)}, \bibinfo{address}{Santa Clara, CA}. pp.
  \bibinfo{pages}{285--300}.
\bibitem[{Low et~al.(2012)Low, Bickson, Gonzalez, Guestrin, Kyrola and
  Hellerstein}]{Low-VLDBE12}
\bibinfo{author}{Low, Y.}, \bibinfo{author}{Bickson, D.},
  \bibinfo{author}{Gonzalez, J.}, \bibinfo{author}{Guestrin, C.},
  \bibinfo{author}{Kyrola, A.}, \bibinfo{author}{Hellerstein, J.M.},
  \bibinfo{year}{2012}.
\newblock \bibinfo{title}{Distributed graphlab: A framework for machine
  learning and data mining in the cloud}.
\newblock \bibinfo{journal}{Proc. VLDB Endow.} \bibinfo{volume}{5},
  \bibinfo{pages}{716–727}.
\newblock \DOIprefix\doi{10.14778/2212351.2212354}.
\bibitem[{Luo et~al.(2021)Luo, Li and Hampshire}]{Luo-JTL21}
\bibinfo{author}{Luo, Q.}, \bibinfo{author}{Li, S.},
  \bibinfo{author}{Hampshire, R.C.}, \bibinfo{year}{2021}.
\newblock \bibinfo{title}{Optimal design of intermodal mobility networks under
  uncertainty: Connecting micromobility with mobility-on-demand transit}.
\newblock \bibinfo{journal}{EURO Journal on Transportation and Logistics}
  \bibinfo{volume}{10}, \bibinfo{pages}{100045}.
\newblock \DOIprefix\doi{10.1016/j.ejtl.2021.100045}.
\bibitem[{Luo et~al.(2022)Luo, Nagarajan, Sundt, Yin, Vincent and
  Shahabi}]{Luo-arXiv22}
\bibinfo{author}{Luo, Q.}, \bibinfo{author}{Nagarajan, V.},
  \bibinfo{author}{Sundt, A.}, \bibinfo{author}{Yin, Y.},
  \bibinfo{author}{Vincent, J.}, \bibinfo{author}{Shahabi, M.},
  \bibinfo{year}{2022}.
\newblock \bibinfo{title}{Efficient algorithms for stochastic ridepooling
  assignment with mixed fleets}.
\newblock \href{http://arxiv.org/abs/2108.08651}{{\tt arXiv:2108.08651}}.
\bibitem[{Ma(2017)}]{Ma-EEEIC17}
\bibinfo{author}{Ma, T.Y.}, \bibinfo{year}{2017}.
\newblock \bibinfo{title}{On-demand dynamic bi-/multi-modal ride-sharing using
  optimal passenger-vehicle assignments}, in: \bibinfo{booktitle}{2017 IEEE
  International Conference on Environment and Electrical Engineering and 2017
  IEEE Industrial and Commercial Power Systems Europe (EEEIC / I CPS Europe)},
  pp. \bibinfo{pages}{1--5}.
\newblock \DOIprefix\doi{10.1109/EEEIC.2017.7977646}.
\bibitem[{Ma et~al.(2019)Ma, Rasulkhani, Chow and Klein}]{Ma-TRELTR19}
\bibinfo{author}{Ma, T.Y.}, \bibinfo{author}{Rasulkhani, S.},
  \bibinfo{author}{Chow, J.Y.}, \bibinfo{author}{Klein, S.},
  \bibinfo{year}{2019}.
\newblock \bibinfo{title}{A dynamic ridesharing dispatch and idle vehicle
  repositioning strategy with integrated transit transfers}.
\newblock \bibinfo{journal}{Transportation Research Part E: Logistics and
  Transportation Review} \bibinfo{volume}{128}, \bibinfo{pages}{417--442}.
\newblock \DOIprefix\doi{10.1016/j.tre.2019.07.002}.
\bibitem[{Masoud et~al.(2017)Masoud, Nam, Yu and Jayakrishnan}]{Masoud-TRR17}
\bibinfo{author}{Masoud, N.}, \bibinfo{author}{Nam, D.}, \bibinfo{author}{Yu,
  J.}, \bibinfo{author}{Jayakrishnan, R.}, \bibinfo{year}{2017}.
\newblock \bibinfo{title}{Promoting peer-to-peer ridesharing services as
  transit system feeders}.
\newblock \bibinfo{journal}{Transportation Research Record}
  \bibinfo{volume}{2650}, \bibinfo{pages}{74--83}.
\newblock \DOIprefix\doi{10.3141/2650-09}.
\bibitem[{Molenbruch et~al.(2021)Molenbruch, Braekers, Hirsch and
  Oberscheider}]{Molenbruch-EJOR21}
\bibinfo{author}{Molenbruch, Y.}, \bibinfo{author}{Braekers, K.},
  \bibinfo{author}{Hirsch, P.}, \bibinfo{author}{Oberscheider, M.},
  \bibinfo{year}{2021}.
\newblock \bibinfo{title}{Analyzing the benefits of an integrated mobility
  system using a matheuristic routing algorithm}.
\newblock \bibinfo{journal}{European Journal of Operational Research}
  \bibinfo{volume}{290}, \bibinfo{pages}{81--98}.
\newblock \DOIprefix\doi{10.1016/j.ejor.2020.07.060}.
\bibitem[{Mourad et~al.(2019)Mourad, Puchinger and Chu}]{Mourad-TRBM19}
\bibinfo{author}{Mourad, A.}, \bibinfo{author}{Puchinger, J.},
  \bibinfo{author}{Chu, C.}, \bibinfo{year}{2019}.
\newblock \bibinfo{title}{A survey of models and algorithms for optimizing
  shared mobility}.
\newblock \bibinfo{journal}{Transportation Research Part B: Methodological}
  \bibinfo{volume}{123}, \bibinfo{pages}{323--346}.
\newblock \DOIprefix\doi{10.1016/j.trb.2019.02.003}.
\bibitem[{Narayan et~al.(2020)Narayan, Cats, van Oort and
  Hoogendoorn}]{Narayan-TRCET20}
\bibinfo{author}{Narayan, J.}, \bibinfo{author}{Cats, O.}, \bibinfo{author}{van
  Oort, N.}, \bibinfo{author}{Hoogendoorn, S.}, \bibinfo{year}{2020}.
\newblock \bibinfo{title}{Integrated route choice and assignment model for
  fixed and flexible public transport systems}.
\newblock \bibinfo{journal}{Transportation Research Part C: Emerging
  Technologies} \bibinfo{volume}{115}.
\newblock \DOIprefix\doi{10.1016/j.trc.2020.102631}.
\bibitem[{Raghunathan et~al.(2018)Raghunathan, Bergman, Hooker, Serra and
  Kobori}]{Raghunathan-SMT18}
\bibinfo{author}{Raghunathan, A.U.}, \bibinfo{author}{Bergman, D.},
  \bibinfo{author}{Hooker, J.}, \bibinfo{author}{Serra, T.},
  \bibinfo{author}{Kobori, S.}, \bibinfo{year}{2018}.
\newblock \bibinfo{title}{Seamless multimodal transportation scheduling}.
\newblock \href{http://arxiv.org/abs/1807.09676}{{\tt arXiv:1807.09676}}.
\bibitem[{Salazar et~al.(2020)Salazar, Lanzetti, Rossi, Schiffer and
  Pavone}]{Salazar-TITS20}
\bibinfo{author}{Salazar, M.}, \bibinfo{author}{Lanzetti, N.},
  \bibinfo{author}{Rossi, F.}, \bibinfo{author}{Schiffer, M.},
  \bibinfo{author}{Pavone, M.}, \bibinfo{year}{2020}.
\newblock \bibinfo{title}{Intermodal autonomous mobility-on-demand}.
\newblock \bibinfo{journal}{IEEE Transactions on Intelligent Transportation
  Systems} \bibinfo{volume}{21}, \bibinfo{pages}{3946--3960}.
\newblock \DOIprefix\doi{10.1109/TITS.2019.2950720}.
\bibitem[{Santi et~al.(2014)Santi, Resta, Szell, Sobolevsky, Strogatz and
  Ratti}]{Santi-PNAS14}
\bibinfo{author}{Santi, P.}, \bibinfo{author}{Resta, G.},
  \bibinfo{author}{Szell, M.}, \bibinfo{author}{Sobolevsky, S.},
  \bibinfo{author}{Strogatz, S.H.}, \bibinfo{author}{Ratti, C.},
  \bibinfo{year}{2014}.
\newblock \bibinfo{title}{Quantifying the benefits of vehicle pooling with
  shareability networks}.
\newblock \bibinfo{journal}{Proceedings of the National Academy of Sciences}
  \bibinfo{volume}{111}, \bibinfo{pages}{13290--13294}.
\newblock \DOIprefix\doi{10.1073/pnas.1403657111}.
\bibitem[{Santos et~al.(2011)Santos, McGuckin, Nakamoto, Gray and
  Liss}]{Santos-USDTFHA11}
\bibinfo{author}{Santos, A.}, \bibinfo{author}{McGuckin, N.},
  \bibinfo{author}{Nakamoto, H.}, \bibinfo{author}{Gray, D.},
  \bibinfo{author}{Liss, S.}, \bibinfo{year}{2011}.
\newblock \bibinfo{title}{Summary of travel trends: 2009 national household
  travel survey}.
\newblock \bibinfo{type}{Technical Report}. US Department of Transportation
  Federal Highway Administration.
\bibitem[{Sierpi{\'{n}}ski(2013)}]{Sierpinski-ATST13}
\bibinfo{author}{Sierpi{\'{n}}ski, G.}, \bibinfo{year}{2013}.
\newblock \bibinfo{title}{Changes of the modal split of traffic in europe}.
\newblock \bibinfo{journal}{Archives of Transport System Telematics}
  \bibinfo{volume}{6}, \bibinfo{pages}{45--48}.
\bibitem[{{Statistics Canada}(2016)}]{StatsCan-2016}
\bibinfo{author}{{Statistics Canada}}, \bibinfo{year}{2016}.
\newblock \bibinfo{title}{Census: Main mode of commuting}.
\bibitem[{Stiglic et~al.(2015)Stiglic, Agatz, Savelsbergh and
  Gradisar}]{Stiglic-TRBM15}
\bibinfo{author}{Stiglic, M.}, \bibinfo{author}{Agatz, N.},
  \bibinfo{author}{Savelsbergh, M.}, \bibinfo{author}{Gradisar, M.},
  \bibinfo{year}{2015}.
\newblock \bibinfo{title}{The benefits of meeting points in ride-sharing
  systems}.
\newblock \bibinfo{journal}{Transportation Research Part B: Methodological}
  \bibinfo{volume}{82}, \bibinfo{pages}{36--53}.
\newblock \DOIprefix\doi{10.1016/j.trb.2015.07.025}.
\bibitem[{Stiglic et~al.(2018)Stiglic, Agatz, Savelsbergh and
  Gradisar}]{Stiglic-COR18}
\bibinfo{author}{Stiglic, M.}, \bibinfo{author}{Agatz, N.},
  \bibinfo{author}{Savelsbergh, M.}, \bibinfo{author}{Gradisar, M.},
  \bibinfo{year}{2018}.
\newblock \bibinfo{title}{Enhancing urban mobility: Integrating ride-sharing
  and public transit}.
\newblock \bibinfo{journal}{Computers \& Operations Research}
  \bibinfo{volume}{90}, \bibinfo{pages}{12--21}.
\newblock \DOIprefix\doi{10.1016/j.cor.2017.08.016}.
\bibitem[{Tafreshian et~al.(2020)Tafreshian, Masoud and Yin}]{Tafreshian-SS20}
\bibinfo{author}{Tafreshian, A.}, \bibinfo{author}{Masoud, N.},
  \bibinfo{author}{Yin, Y.}, \bibinfo{year}{2020}.
\newblock \bibinfo{title}{Frontiers in service science: ride matching for
  peer-to-peer ride sharing: a review and future directions}.
\newblock \bibinfo{journal}{Service Science} \bibinfo{volume}{12},
  \bibinfo{pages}{41--60}.
\newblock \DOIprefix\doi{10.1287/serv.2020.0258}.
\bibitem[{Thao et~al.(2021)Thao, Imhof and von Arx}]{Thao-TRIP21}
\bibinfo{author}{Thao, V.T.}, \bibinfo{author}{Imhof, S.}, \bibinfo{author}{von
  Arx, W.}, \bibinfo{year}{2021}.
\newblock \bibinfo{title}{Integration of ridesharing with public transport in
  rural switzerland: practice and outcomes}.
\newblock \bibinfo{journal}{Transportation Research Interdisciplinary
  Perspectives} \bibinfo{volume}{10}, \bibinfo{pages}{100340}.
\newblock \DOIprefix\doi{10.1016/j.trip.2021.100340}.
\bibitem[{Tirachini and Gomez-Lobo(2020)}]{Tirachini-IJST20}
\bibinfo{author}{Tirachini, A.}, \bibinfo{author}{Gomez-Lobo, A.},
  \bibinfo{year}{2020}.
\newblock \bibinfo{title}{Does ride-hailing increase or decrease vehicle
  kilometers traveled ({VKT})? a simulation approach for santiago de chile}.
\newblock \bibinfo{journal}{International Journal of Sustainable
  Transportation} \bibinfo{volume}{14}, \bibinfo{pages}{187--204}.
\newblock \DOIprefix\doi{10.1080/15568318.2018.1539146}.
\bibitem[{Vora(2019)}]{Vora-ATC19}
\bibinfo{author}{Vora, K.}, \bibinfo{year}{2019}.
\newblock \bibinfo{title}{{LUMOS}: {Dependency-Driven} disk-based graph
  processing}, in: \bibinfo{booktitle}{2019 USENIX Annual Technical Conference
  (USENIX ATC 19)}, \bibinfo{address}{Renton, WA}. pp.
  \bibinfo{pages}{429--442}.
\bibitem[{Vora et~al.(2016)Vora, Xu and Gupta}]{Vora-ATC16}
\bibinfo{author}{Vora, K.}, \bibinfo{author}{Xu, G.}, \bibinfo{author}{Gupta,
  R.}, \bibinfo{year}{2016}.
\newblock \bibinfo{title}{Load the edges you need: A generic {I/O} optimization
  for disk-based graph processing}, in: \bibinfo{booktitle}{2016 USENIX Annual
  Technical Conference (USENIX ATC 16)}, \bibinfo{address}{Denver, CO}. pp.
  \bibinfo{pages}{507--522}.
\bibitem[{Wang and Odoni(2014)}]{Wang-TS14}
\bibinfo{author}{Wang, H.}, \bibinfo{author}{Odoni, A.}, \bibinfo{year}{2014}.
\newblock \bibinfo{title}{Approximating the performance of a ``last mile''
  transportation system}.
\newblock \bibinfo{journal}{Transportation Science} \bibinfo{volume}{50},
  \bibinfo{pages}{659--675}.
\newblock \DOIprefix\doi{10.1287/trsc.2014.0553}.
\bibitem[{Wang and Yang(2019)}]{Wang-TRBM19}
\bibinfo{author}{Wang, H.}, \bibinfo{author}{Yang, H.}, \bibinfo{year}{2019}.
\newblock \bibinfo{title}{Ridesourcing systems: A framework and review}.
\newblock \bibinfo{journal}{Transportation Research Part B: Methodological}
  \bibinfo{volume}{129}, \bibinfo{pages}{122--155}.
\newblock \DOIprefix\doi{10.1016/j.trb.2019.07.009}.
\bibitem[{Zhang and Zhang(2018)}]{Zhang-IJERPH18}
\bibinfo{author}{Zhang, Y.}, \bibinfo{author}{Zhang, Y.}, \bibinfo{year}{2018}.
\newblock \bibinfo{title}{Exploring the relationship between ridesharing and
  public transit use in the united states}.
\newblock \bibinfo{journal}{International Journal of Environmental Research and
  Public Health} \bibinfo{volume}{15}, \bibinfo{pages}{1763}.
\newblock \DOIprefix\doi{10.3390/ijerph15081763}.

\end{thebibliography}

\end{document}